\documentclass[11pt,letterpaper,onecolumn,titlepage]{quantumarticle_tweaked}
\pdfoutput=1
\pdfoutput=1
\usepackage{fullpage} 
\usepackage{xspace,xcolor,graphicx,tabularx} 
\usepackage{mathtools,amsthm,amssymb,bm,lmodern} 
\usepackage{enumitem} 
\setenumerate[0]{label={\normalfont (\roman*)}}

\usepackage[T1]{fontenc}
\usepackage[UKenglish]{babel}
\usepackage{dsfont} 
\let\mathbb\relax
\let\mathbb\mathds
\usepackage{stmaryrd} 
\makeatletter
\renewcommand{\paragraph}{%
  \@startsection{paragraph}{4}%
  {\z@}{2.25ex \@plus 1ex \@minus .2ex}{-1em}%
  {\normalfont\normalsize\bfseries}%
}
\makeatother
\definecolor{linkblue}{HTML}{001487}
\usepackage[colorlinks=true,allcolors=linkblue]{hyperref}
\usepackage{url}
\usepackage[nameinlink,capitalize,noabbrev]{cleveref}
\crefname{enumi}{Step}{Steps}

\newtheorem{theorem}{Theorem}[section]
\newtheorem*{theorem*}{Theorem}

\newtheorem{lemma}[theorem]{Lemma}

\newtheorem{corollary}[theorem]{Corollary}
\theoremstyle{remark}
\newtheorem{remark}[theorem]{Remark}
\theoremstyle{definition}
\newtheorem{definition}[theorem]{Definition}
\newtheorem{example}[theorem]{Example}

\numberwithin{equation}{section}
\newcommand\numberthis{\addtocounter{equation}{1}\tag{\theequation}}

\newcommand{\setft}[1]{\textnormal{#1}}
\newcommand{\eps}{\epsilon}
\newcommand{\1}{\mathbb{1}}

\newcommand{\C}{\ensuremath{\mathds{C}}}
\newcommand{\N}{\ensuremath{\mathds{N}}}
\newcommand{\R}{\ensuremath{\mathds{R}}}

\newcommand{\bits}{\ensuremath{\{0, 1\}}}

\newcommand{\ot}{\ensuremath{\otimes}}
\newcommand{\deq}{\coloneqq}
\newcommand{\tr}[1]{\mbox{\rm Tr}\!\left[ #1 \right]}

\newcommand{\pr}[1]{{\rm Pr}\!\left[ #1 \right]}
\newcommand{\prs}[2]{{\rm Pr}_{#1}\!\left[ #2 \right]}
\newcommand{\norm}[1]{\left\lVert#1\right\rVert}
\DeclareMathOperator{\pos}{Pos}

\newcommand{\sth}{{\setft{~s.t.~}}}
\newcommand{\tand}{\;\textnormal{~and~}\;}

\newcommand{\ket}[1]{|#1\rangle}
\newcommand{\bra}[1]{\langle#1|}
\newcommand{\proj}[1]{\ket{#1}\!\bra{#1}}
\DeclarePairedDelimiterX\braket[2]{\langle}{\rangle}{#1 \delimsize\vert #2}

\newcommand{\cD}{\ensuremath{\mathcal{D}}}

\newcommand{\cH}{\ensuremath{\mathcal{H}}}
\newcommand{\cI}{\ensuremath{\mathcal{I}}}

\newcommand{\cN}{\ensuremath{\mathcal{N}}}

\newcommand{\tln}{\textsc{tln}}
\let\i\relax
\newcommand{\i}{\iota}

\newcommand{\pmset}{\{\pm 1\}}

\begin{document}


\date{} 

\title{Concentration bounds for quantum states and limitations on the QAOA from polynomial approximations}

\author{Anurag Anshu}
\affiliation{School of Engineering and Applied Sciences, Harvard University}
\email{anuraganshu@fas.harvard.edu}

\author{Tony Metger}
\affiliation{Institute for Theoretical Physics, ETH Zurich}
\email{tmetger@ethz.ch}

\thanks{Note: a short version of this work has appeared in the \href{https://doi.org/10.4230/LIPIcs.ITCS.2023.5}{Proceedings of the 14th Innovations in Theoretical Computer Science Conference (ITCS 2023).}}
\maketitle

\vspace{-0.7cm}

\begin{abstract}
We prove concentration bounds for the following classes of quantum states: 
(i) output states of shallow quantum circuits, answering an open question from \cite{DMRF22}; 
(ii) injective matrix product states; 
(iii) output states of dense Hamiltonian evolution, i.e.~states of the form $e^{\iota H^{(p)}} \cdots e^{\iota H^{(1)}} \ket{\psi_0}$ for any $n$-qubit product state $\ket{\psi_0}$, where each $H^{(i)}$ can be any local commuting Hamiltonian satisfying a norm constraint, including dense Hamiltonians with interactions between any qubits.
Our proofs use polynomial approximations to show that these states are close to local operators.
This implies that the distribution of the Hamming weight of a computational basis measurement (and of other related observables) concentrates.

An example of (iii) are the states produced by the quantum approximate optimisation algorithm (QAOA). 
Using our concentration results for these states, we show that for a random spin model, the QAOA can only succeed with negligible probability even at super-constant level $p = o(\log \log n)$, assuming a strengthened version of the so-called overlap gap property. 
This gives the first limitations on the QAOA on dense instances at super-constant level, improving upon the recent result~\cite{BGMZ22}.
\end{abstract}
\vspace{0.3cm}

\section{Introduction}

Concentration bounds deal with the deviation of random variables from their expectation.
Such bounds describe important structural properties of probability distributions that carry low correlation, and as a result have become ubiquitous tools in mathematics, computer science, and physics.
A series of recent works has extended concentration bounds to the quantum many-body setting, where weakly correlated quantum states hold physical and computational relevance and their concentration properties explain important physical effects.
For example, the concentration of expectation values in various quantum states such as product states~\cite{GV89, HGH04, Ab20}, Gibbs quantum states~\cite{DR22, KS20}, and finitely correlated states~\cite{A16} explains the equivalence of ensembles~\cite{BC15, BCG15, T18, Al22} and eigenstate thermalisation~\cite{KS20b} in quantum statistical mechanics. 
Concentration bounds are also an important proof technique in quantum complexity theory: for example, concentration properties of quantum states generated by low-depth quantum circuits have been used to prove circuit lower bounds on low-energy states of quantum Hamiltonians~\cite{EldarH17}, leading to the recent proof of the NLTS conjecture~\cite{ANirkhe22, AB22, ABN22}. 
Furthermore, concentration results play an important role in analysing and bounding the performance of variational quantum algorithms for classical constraint optimisation problems, e.g.~the quantum approximate optimisation algorithm (QAOA)~\cite{farhi2014quantum,farhi2020quantum,CLSS22, BGMZ22}.

The standard method to prove concentration inequalities is the moment method. 
The slow growth of the moments of weakly correlated probability distributions or quantum states serves as a signature of concentration.
Arguably the simplest example is the classical Chernoff-Hoeffding bound, which shows that the probability that the sum of $n$ independent random variables deviates from its expectation value by more than $k$ is at most $e^{-\Omega(k^2/n)}$.
We call this \emph{Gaussian concentration}.
The moment method also extends to non-commuting observables with product quantum states~\cite{K16} or quantum states with exponential decay of correlation~\cite{ A16}, albeit sometimes with weaker bounds of the form $e^{-\Omega(k^\alpha/n^\beta)}$ for some $\alpha, \beta > 0$.
We call these weaker bounds \emph{exponential concentration}.\footnote{Note that with this definition, technically Gaussian concentration is a special case of exponential concentration for $\alpha = 2$ and $\beta = 1$. However, when we say ``exponential concentration'' it is usually implicit that $\alpha$ and $\beta$ are such that the bounds are generally weaker than for Gaussian concentration. A typical case is $\alpha = 1$ and $\beta = 1/2$.}

An alternative to the moment method was introduced in~\cite{KAAV15}.
The idea of this approach is to approximate the quantum state of interest by a local operator and use this approximation to prove concentration bounds.
Because the local approximation is usually constructed from low-degree polynomials, we call this the \emph{polynomial-based method}.
The strength of the resulting concentration bounds depends on the locality of the approximation; see \cref{sec:concfromloc} for details.
\cite{KAAV15} used this method to show exponential concentration of the form $e^{-\Omega(k/\sqrt{n})}$ for local classical observables on the ground states of gapped local Hamiltonians. 
However, their method was not able to produce Gaussian concentration even for the simplest case of local observables on an i.i.d.~distribution (e.g.~the sum of i.i.d.~random variables), a case in which the Chernoff-Hoeffding bound does give Gaussian concentration results.

\subsection{Main results}
We extend the polynomial-based method in two ways: we show that in cases where the moment method produces Gaussian concentration, the polynomial-based method can do so, too; and we show that the polynomial-based method can be applied to a much wider class of states, including the output states of shallow quantum circuits, injective matrix product states, and the output states of dense Hamiltonian evolutions (explained below).
An example of a dense Hamiltonian evolution is the QAOA for solving classical constraint optimisation problems (COPs).
We therefore obtain concentration bounds for the QAOA, and combining this with the so-called \emph{overlap gap property} first introduced in the classical literature~\cite{gamarnik2018finding,gamarnik2021overlap}, we prove strong limitations on the performance of the QAOA even at (admittedly only slightly) super-constant level $p = o(\log\log n)$.
Crucially, our method works for \emph{dense} COPs, which may have constraints between any variables, and our proofs are fairly straightforward.
This improves upon the recent work~\cite{BGMZ22}, which proved similar results for constant-level QAOA on dense instances using a highly technical proof.

We now describe our main results in more detail. In all cases, concentration bounds are obtained by approximating the quantum state of interest by a local operator, so in sketching the proof ideas, we only focus on the construction of such a local approximation.
Below, we only give concentration bounds for the Hamming weight distribution $W_{\rho}$ of an $n$-qubit quantum state $\rho$, i.e.~the probability distribution over $\{1, \dots, n\}$ describing the Hamming weight of a computational basis measurement of $\rho$. More formally, $\pr{W_\rho = i} = \sum_{x \in \bits^n : |x| = i} \bra{x} \rho \ket{x}$.
However, it is straightforward to extend these bounds to other observables that only change slowly as the Hamming weight changes; for an example, see~\cref{lem:qaoa_energy_conc}.

\begin{table}[ht!]
\centering
{\renewcommand{\arraystretch}{1.6}
 \begin{tabular}{|m{5cm} | m{5cm} | m{5cm}|} 
\hline
 & Prior work & This work\\ 
\hline \hline
Depth $t$ quantum states & $e^{-\Omega(k/\sqrt{2^{2t} n})}$ (implicit in \cite{KAAV15}) & $e^{-\Omega(k^2/2^{2t} n)}$\\
\hline
Injective matrix product states & $e^{-\Omega(k/\sqrt{n})}$ (\cite{KAAV15, A16}) & $e^{-\Omega(k^2/ n)}$\\
\hline
Dense Hamiltonian evolution  & $o_n(1)$ for the special case of QAOA with level $p = O(1)$~\cite{BGMZ22} & $e^{-\Omega(n^{1/8})}$ for $k=o(n)$ and level $p=o(\log\log n)$\\ 
\hline
 \end{tabular}}
\caption{Main results and comparison with prior work. 
For each case, we consider a state on $n$ qubits and bound the probability that the Hamming weight of a computational basis measurements deviates by more than $k \in [0,n]$ from its median (or mean).
Note that~\cite{BGMZ22} were able to show concentration over both the choice of instance and the randomness of the QAOA output, whereas our bounds, while stronger, only deal with concentration over the latter.}
\label{tab:probtoquant}
\end{table}

\paragraph{Shallow quantum circuits.}
Consider a depth-$t$ quantum circuit, i.e.~a circuit comprised of $t$ layers of arbitrary 2-qubit gates applied to the initial state $\ket{0}^{\ot n}$.
We denote the unitary implemented by this circuit by $U$.
The output state of this circuit is the unique maximum-energy eigenstate of a $2^t$-local Hamiltonian that is a sum of commuting projectors.
To see this, observe that $U \ket{0}^{\ot n}$ is the unique joint $(+1)$-eigenstate of the operators $H_i = U \proj{0}_i U^\dagger$, where $\proj{0}_i$ acts as identity on all qubits except $i$.
By a standard lightcone argument, $H_i$ only acts non-trivially on $2^t$ qubits, so $H = \frac{1}{n}\sum H_i$ is a $2^t$-local Hamiltonian with $U \ket{0}^{\ot n}$ as its unique $(+1)$-eigenstate.
Therefore, the output of the circuit can be written as $U\proj{0}^{\ot n}U^\dagger = \delta_1(H)$, where $\delta_1(1) = 1$ and $\delta_1(x) = 0$ for $x \neq 1$.\
We can now approximate $\delta_1(x)$ using a degree-$d$ polynomial $P_d$ constructed in~\cite{KahnLS96, BuhrmanCWZ99, AAG21}.
Polynomials spread the locality of operators in a controllable way.
We can therefore show that $P_d(H)$ is a $(d \cdot 2^t)$-local operator that approximates $U\proj{0}^{\ot n} U^\dagger$. 
Hence, we have constructed a local operator approximation to $U \proj{0}^{\ot n} U^\dagger$ and can use this approximation to show that for any depth-$t$ circuit, the output state $\proj{\psi} = U \proj{0}^{\ot n} U^\dagger$ has the following concentration property for $k \in (2^t\sqrt{n}, 2^tn)$ (see \cref{lem:lowdepth} for the formal statement):
\begin{align*}
\pr{|W_\psi - \textnormal{median}(W_\psi)| \geq k} \leq e^{- \Omega\left( \frac{k^2}{2^{2t} n} \right)} \,.
\end{align*}
This generalises the Chernoff-Hoeffding bound for product distributions (which corresponds to the case of a single layer ($t = 1$), since then the measurement distribution of the output state is a product of distributions on one or two bits) and shows Gaussian concentration for any constant-depth quantum circuit, answering an open question from~\cite{DMRF22}. Similar statements have also appeared in~\cite{ANirkhe22, AB22, ABN22}. 
Note that the same argument applies to any ground state of a Hamiltonian that is a sum of commuting projectors, not just output states of shallow circuits.

\paragraph{Injective matrix product states.}
Matrix product states (MPSs) are a widely used tensor network representation of quantum states.
Injective MPSs have an additional property that ensures that they are the unique ground state of a local ``parent Hamiltonian'' with a constant spectral gap.
We can therefore approximate an injective MPS as a polynomial of its parent Hamiltonian.
Using near-optimal polynomial approximations constructed in~\cite{AAG21}, we obtain Gaussian concentration bounds for injective MPSs (\cref{lem:mps_bound}). 
Our bounds are stronger than previous ones~\cite{A16, KAAV15}, which only showed exponential concentration. 
We also note that conditionally independent probability distributions can be encoded into injective MPSs, and that in that case our concentration bounds reproduce a (version of) Azuma's inequality.

\paragraph{Dense Hamiltonian evolution.}
Concentration bounds are natural for quantum states that have weak long-range correlations such as the output states of shallow quantum circuits.
A priori, one would not expect similar bounds to hold for quantum states with long-range correlations. 
Recently, \cite{BGMZ22} considered the output distribution of the QAOA (explained below) on random dense COPs, i.e.~local COPs that can have constraints between any variables.
For dense COPs, the operations implemented by the QAOA can include interactions between any qubits, and as a result the output distribution can have long-range correlations.
Remarkably,~\cite{BGMZ22} showed that the variance of the average energy density (averaged over the randomness of the QAOA as well as the choice of random instance) vanishes asymptotically.
This means that the energy density of the output, which corresponds to the quality of the COP solution produced by the QAOA, concentrates about the average. 
However,~\cite{BGMZ22} were only able to prove an asymptotic statement without explicit tail bounds and the proof was highly non-trivial.

We consider a more general class of states that includes the output states of the QAOA as a special case.
Specifically, we define the output of a \emph{dense Hamiltonian evolution} as a state of the form $e^{\iota H^{(p)}} \cdots e^{\iota H^{(1)}} \ket{\psi_0}$ for any $n$-qubit product state $\ket{\psi_0}$.
Here, each $H^{(i)}$ can be any commuting local Hamiltonian (though the different $H^{(i)}$ themselves are of course not required to commute).
Importantly, $H^{(i)}$ are allowed to be \emph{dense} Hamiltonians, i.e.~Hamiltonians with interactions between any qubits.
As explained below, the QAOA applied to a dense COP is a special case of dense Hamiltonian evolution.

For our concentration bounds, we further require each $H^{(i)}$ to satisfy a norm constraint, which limits the norm of the Hamiltonian restricted to a subset of the qubits (see \cref{eqn:subset_condition} for the formal statement).
In particular, this condition is satisfied with overwhelming probability for the random dense model from~\cite{BGMZ22}.
Under this condition, we can prove (\cref{thm:qaoamain}) that the output state of a dense Hamiltonian evolution is $\eps$-close in operator norm to a $k_p$-local operator for 
\begin{align*}
k_p \leq c_1^p n^{1 - (1-\alpha)^p/4}\,,\quad
\eps \leq e^{-\Omega(n^{1/8})} \,.
\end{align*}
Here, $c_1$ and $0 \leq \alpha < 1$ are constants and $p$ is \emph{level} of the dense evolution, i.e.~the number of unitaries $e^{\iota H^{(i)}}$ that have been applied.
In particular, if we choose $p = o(\log \log n)$, then $k = o(n)$.
This implies the following exponential concentration result (\cref{lem:qaoa_weight_conc}): for $\rho_p$ the output of a dense Hamiltonian evolution with level $p = o(\log\log n)$ satisfying the above conditions,
\begin{align*}
\pr{|W_{\rho_p}-\textnormal{median}(W_{\rho_p})| > o(n)} \leq e^{-\Omega(n^{1/8})}\,.
\end{align*}
Here, we have only stated the asymptotic result, but in \cref{lem:qaoa_weight_conc} we give explicit bounds for any choice of $p$.
This concentration result can also be extended beyond just the Hamming weight of a computational basis measurement: for example, it also holds for the energy density of $\rho_p$ with respect to any classical Hamiltonian satisfying a similar norm constraint to the one mentioned above (\cref{lem:qaoa_energy_conc}).

To prove that the output of a dense Hamiltonian evolution can  be approximated by a local operator, we again make use of polynomial approximations.
Recall that we are interested in states of the form $\rho_p = e^{\i H^{(p)}} \cdots e^{\i H^{(1)}} \proj{\psi_0} e^{-\i H^{(1)}} \cdots e^{-\i H^{(p)}}$ for a pure product state $\ket{\psi_0}$.
As a first step, we approximate $\proj{\psi_0}$ by a local operator (\cref{lem:starting_state_approx}).
For this, we observe that since $\ket{\psi_0} = \otimes_i \ket{\psi_0}_i$ is a product state, it is the unique ground state of the $1$-local Hamiltonian $H = \frac{1}{n} \sum \proj{\psi_0}_i$.
If we apply a linear combination of Chebyshev polynomials to this Hamiltonian, we obtain a good local approximation to $\proj{\psi_0}$.
Then, for each unitary $e^{\i H^{(i)}}$ in the dense Hamiltonian evolution, we approximate the exponential function by its truncated Taylor series.
The fact that the Hamiltonian evolution is applied to an approximately local operator allows us to use the norm constraint mentioned above to obtain an improved error bound for the truncated Taylor series (see \cref{lem:comm_locality_spread} for details). 
Therefore, the truncated Taylor series spreads the locality of the state in a controllable way without degrading the quality of the approximation too much.
Applying this argument recursively for each layer of the dense Hamiltonian evolution, we obtain a local approximation to the output state $\rho_p$.

\paragraph{Limitations on the QAOA from concentration bounds.}
The QAOA~\cite{farhi2014quantum} is an algorithm for solving local COPs (i.e.~COPs consisting of any number of clauses, each with at most $q=O(1)$ variables) on a quantum computer.
We can associate a $q$-local Hamiltonian $H$ with every $q$-local COP $C$ by replacing the variables in $C$ with Pauli-Z matrices acting on different qubits.
The resulting Hamiltonian is diagonal in the computational basis and has the property that for any string $x \in \bits^n$, $C(x) = \bra{x} H \ket{x}$.
The QAOA attempts to find a ``good'' solution $x$ (i.e.~one for which $C(x)$ is as large as possible) by starting from the state $\ket{+}^{\ot n}$ and then applying $p$ layers of unitaries of the form $e^{\i \beta_i \sigma_X^{\ot n}} e^{\i \gamma_i H}$. Here, $\beta_i$ and $\gamma_i$ are real parameters that can be tuned to the problem instance. 
It is clear that this is a special case of the dense Hamiltonian evolution we have described earlier.
Our results will apply for any choice of $\beta_i$ and $\gamma_i$ and we will always implicitly consider a family of COPs, one for each number $n$ of input bits, in order to make asymptotic statements.

\cite{BGMZ22} considered the performance of the QAOA on a random spin model on $n$ qubits, described by the $q$-local Hamiltonian 
\begin{align}
H_n^q(J) = \frac{1}{n^{(q-1)/2}} \sum_{i_1,\ldots i_q=1}^n J_{i_1,\ldots i_q} \sigma^Z_{i_1}\ldots \sigma^Z_{i_q} \,, \label{eqn_intro_spin_model}
\end{align}
where $J_{i_1,\ldots i_q} \sim \cN(0,1)$ are sampled from i.i.d.~standard Gaussians.
For this model, \cite{BGMZ22} were able to show that for constant even $q \geq 4$ and level $p = O(1)$, the value achieved by the QAOA (for fixed $\beta_i, \gamma_i$) in expectation over $J$ and the internal randomness of the QAOA is bounded away from the optimal value by a constant as $n \to \infty$.
They were also able to show the asymptotic concentration property described above.

Here, we use our concentration results to show limitations on the QAOA for a class of COPs that includes the random spin model above.
For this, we consider local COPs that have the so-called overlap gap property (OGP)~\cite{gamarnik2021survey}, which roughly says that ``good'' solutions to the COP are clustered in the sense that two good solutions are either close or far in Hamming distance.
Combining this with our concentration results for dense Hamiltonian evolution, we can show that for any COP with a sufficiently strong OGP whose associated Hamiltonian satisfies \cref{eqn:subset_condition}, if the QAOA produced a good solution with noticeable probability, then the probability distribution over good solutions produced by the QAOA would have to be concentrated on one such cluster.
This allows us to show that the QAOA cannot succeed with noticeable probability on \emph{symmetric} COPs (i.e.~COPs that are invariant under flipping all the input bits) that have a strong OGP.
This is because the symmetry of the COP is in contradiction with the existence of a single cluster on which most of the probability distribution is concentrated: if such a cluster existed, we could take the strings in that cluster and flip all their bits to produce another cluster which, by symmetry, must have the same probability weight, a contradiction.
This argument is similar to~\cite{BravyiKKT19}.
As a result, we obtain the following limitation on the QAOA (see~\cref{lem:symm_qaoa_limit}).

\begin{theorem*}[informal]
Consider a local symmetric COP $C(x)$ with a sufficiently strong OGP and suppose that the associated Hamiltonian $H$ satisfies the norm constraint in \cref{eqn:subset_condition}.
Then, the value of the solution to $C(x)$ produced by the QAOA with level $p = o(\log\log n)$ is bounded away from the optimal value by at least a constant except with probability $e^{-\Omega(n^{1/8})}$.
\end{theorem*}

For even $q$, the random spin model from~\cref{eqn_intro_spin_model} is symmetric and satisfies the norm constraint in \cref{eqn:subset_condition} with overwhelming probability.
Furthermore, it was shown in~\cite{CGPR19} that it satisfies the OGP with overwhelming probability.
However, we note that here we need a stronger version of the OGP than was shown in~\cite{CGPR19}.
This stronger version appears to be implicit in their proof, too, although we leave its formal proof for future work and assume it here as a conjecture.
Assuming this stronger OGP, we can show the following (see~\cref{sec:symm_qaoa_limitations}).
\begin{theorem*}[informal]
With probability $1 - O(e^{-n})$ over the choice of $J$ (with i.i.d.~Gaussian entries), the value of the solution to the random spin model (\cref{eqn_intro_spin_model}) produced by the QAOA with level $p = o(\log\log n)$ is bounded away from the optimal value by at least a constant except with probability $e^{-\Omega(n^{1/8})}$.
\end{theorem*}
This places strong limitations on the performance of the QAOA because it does not just bound the \emph{expectation} value away from the optimal value as in~\cite{BGMZ22}, but instead asserts that the QAOA output is bounded away from the optimal value with overwhelming probability, even at super-constant level $p = o(\log \log n)$.

\subsection{Discussion and open questions}
We have shown that polynomial approximations can be used to derive Gaussian concentration bounds for the output states  of constant-depth quantum circuits and injective matrix product states, and exponential concentration bounds for the output states of dense Hamiltonian evolution.
The latter can be used  to derive strong limitations on the performance of the QAOA at super-constant level $p = o(\log \log n)$ even on dense instances such as random spin models.

At first sight, it is surprising that the (provably optimal) polynomial approximations~\cite{KahnLS96, BuhrmanCWZ99} we use for shallow quantum circuits are able to reproduce the (likewise provably optimal) Chernoff-Hoeffding bound in the classical case.
It would be interesting to explore whether there is a deeper conceptual connection between optimal polynomial approximations and optimal concentration bounds.

On a more technical level, there are a number of interesting improvements one could hope to make to our bounds.
Firstly, our bounds for MPSs (\cref{lem:mps_bound}) can only deal with sub-linear deviations $k = O(n^{1 - \delta})$ for any $\delta > 0$. It would be desirable to extend this result to arbitrary values of $k$.
Additionally, one could hope to prove similar concentration bounds for PEPSs, the two-dimensional analogue of MPSs.

Secondly, we only achieve exponential, not Gaussian, concentration bounds for dense Hamiltonian evolutions with level $p=o(\log \log n)$.
Can one improve these results to Gaussian concentration and also extend them to higher levels, e.g.~$p = O(\log n)$ or even $p = O(n^{\delta})$ for a small $\delta > 0$?
Furthermore, we show concentration for the output states of dense Hamiltonian evolution for a fixed instance, but we cannot show that for random COPs, the output states also have concentration properties over the choice of random instance, e.g.~over the choice of $J_{i_1 \dots i_q} \sim \cN(0,1)$ in the case of the random spin model introduced earlier.
\cite{BGMZ22} do show such a concentration property, albeit only in the asymptotic regime without explicit bounds.
Can our polynomial approximation techniques also be used to prove explicit concentration bounds over the choice of random instance?
If so, it might be possible to extend the limitations on the performance of dense evolutions for COPs we prove in \cref{subsec:symopt} beyond symmetric COPs and optimisers.

Finally, our techniques may also be useful for problems in condensed matter physics.
As an example, consider the Lieb-Schultz-Mattis theorem \cite{LiebSM61} and its higher-dimensional generalisation \cite{Hastings04}, seminal results in condensed matter physics. 
Their main idea is that sufficient symmetry and non-degeneracy of the ground space prevents a Hamiltonian from being gapped. 
Inspired by our application to symmetric QAOA (\cref{sec:symm_qaoa_limitations}), we can ask whether an alternative proof of this result can be obtained using concentration bounds and polynomial approximations, e.g.~by showing that the concentration properties of unique gapped ground states are in conflict with the symmetry requirements.  

\paragraph{Acknowledgements.}
We thank Joao Basso, David Gamarnik, Song Mei, and Leo Zhou for very helpful discussions and especially for suggesting the application to symmetric QAOA (\cref{sec:symm_qaoa_limitations}). 
AA also thanks Daniel Stilck Fran{\c{c}}a, Tomotaka Kuwahara, Cambyse Rouz{\'e}, and Juspreet Singh Sandhu for helpful discussions. 
This work was done in part while the authors were visiting the Simons Institute for the Theory of Computing. 
AA acknowledges support through the NSF CAREER Award No. 2238836 and NSF award QCIS-FF: Quantum Computing \& Information Science Faculty Fellow at Harvard University (NSF 2013303).
TM acknowledges support from the ETH Zurich Quantum Center.

\section{Preliminaries}
\subsection{Notation}
For a bitstring $x \in \bits^n$, we denote by $(-1)^x$ the string $((-1)^{x_1}, \dots, (-1)^{x_n})$. For $n \in \N$, $[n] \deq \{1, \dots, n\}$.
We consider the $n$-qubit Hilbert space $\cH = (\C^2)^n$ throughout.
$L(\cH)$ is the set of linear operators on $\cH$, and $\cD(\cH) = \{A \in \pos(\cH)\,|\; \tr{A} = 1\}$ is the set of density matrices on $\cH$.
For a subset $S \in \bits^n$, we use $\Pi_S$ to denote the projector onto $S$, i.e.~$\Pi_S = \sum_{x \in S} \proj{x}$.
The single qubit Pauli operators are $\sigma_X = \left( \begin{smallmatrix} 0 & 1 \\ 1 & 0 \end{smallmatrix}  \right)$ and $\sigma_Z = \left( \begin{smallmatrix} 1 & 0 \\ 0 & -1 \end{smallmatrix}  \right)$.

For a vector $\ket{\psi} \in \cH$ with entries $\psi_i$ and $p \in [1, \infty]$, the $p$-norm is defined as $\norm{\ket{\psi}}_p = \left( \sum_{i} |\psi_i|^p \right)^{1/p}$.
For an operator $A \in L(\cH)$ and $p \in [1,\infty]$, the Schatten $p$-norm $\norm{A}_p$ is the (vector) $p$-norm of the vector of singular values of $A$.
For $p = \infty$, this norm is called the operator norm and can also be written as $\norm{A}_\infty = \sup_{\norm{\ket{\psi}}_2 = 1} \norm{A \ket{\psi}}_2$.
For convenience, we usually drop the subscript for the operator norm, i.e.~$\norm{A} \deq \norm{A}_\infty$.

\begin{definition}[Hamming weight distribution]
For a quantum state $\rho \in \cD(\cH)$ we denote by $W_\rho$ the random variable indicating the Hamming weight of the outcome of measuring $\rho$ in the computational basis, i.e. 
\begin{align*}
\pr{W_\rho = i} = \sum_{x \in \bits^n : |x| = i} \bra{x} \rho \ket{x} \,.
\end{align*}
\end{definition}

\subsection{Local operators}
\begin{definition}[Local operators]
Let $k \in \N$ and $\eps > 0$.
An operator $R \in L(\cH)$ is called $k$-local if it can be written as $R = \sum_{i} R_i$ for operators $R_i$ that only act non-trivially on $k$ subsystems.
Whenever we write $R = \sum_{i} R_i$ for a $k$-local operator $R$, this is understood to be such a local decomposition.
An operator $Q \in L(\cH)$ is called $(k, \eps)$-local if there exists a $k$-local operator $R$ such that $\norm{Q - R} \leq \eps$.
\end{definition}

We will frequently consider local operators with additional properties and extend the above definition in the obvious way: for example, a $(k, \eps)$-local state is a quantum state that is $\eps$-close in operator norm to a $k$-local operator $R$.
Note that the operator $R$ does not need to be a quantum state itself.

\begin{definition}[Total local norm] \label{def:tln}
A $k$-local operator $R$ has total local norm $\tln(R) \leq r$ if there exists a local decomposition $R = \sum R_i$ such that $\sum \norm{R_i} \leq r$.
Similarly, a $(k, \eps)$-local operator $Q$ has approximate total local norm $\tln_\eps(Q) \leq q$ if there exists a $k$-local operator $R$ within $\eps$-distance (in operator norm) from $Q$ with $\tln(R) \leq q$.
\end{definition}
Note that in the definition of $\tln(R)$ and $\tln_\eps(Q)$, it always has to be clear from the context which locality $k$ we are considering, since for different choices of $k$, different values of the total local norm can be achieved.
Therefore we will always make statements of the form: $R$ is a $k$-local operator with $\tln(R) = \dots$ (and likewise for the approximate case).
Hence, strictly speaking the subscript in $\tln_\eps(Q)$ is unnecessary as it must anyway be specified that $Q$ is a $(k, \eps)$-local operator for some values of $k$ and $\eps$, but we find it useful to include the subscript as a reminder of this nonetheless.

\begin{example} \label{ex:poly_to_local_op}
The locality of an operator increases in a controllable way when we apply a polynomial.
Specifically, consider a $k$-local operator $R = \sum R_i$ and a degree-$d$ polynomial $P$.
Then, $P(R)$ is a $(d \cdot k)$-local operator because we can expand it into a sum of terms, each of which contains a product of at most $d$ different local terms $R_i$.
\end{example}

\begin{example}
\label{ex:plusstate}
The state $\proj{+}^{\otimes n}$ is a non-local operator. 
However, it can be approximated by a local operator.
For this, consider the Hamiltonian 
$H_0 = \frac{1}{n} \sum_{i = 1}^n \proj{+}_i$, where $\proj{+}_i$ acts as identity everywhere except on the $i$-th qubit.
$H_0$ is a 1-local operator, with ground state $\proj{+}^{\ot n}$ and spectral gap $1/n$.
It is easy to see that $\|\proj{+}^{\otimes n} - H_0^k\|\leq (1-\frac{1}{n})^k \leq e^{-k/n}$, and by \cref{ex:poly_to_local_op} $H_0^k$ is a $k$-local operator.
We therefore have shown that $\proj{+}^{\ot n}$ is a $\left(k, (1-\frac{1}{n})^k\right)$-local operator.
Expanding $H_0^k=\frac{1}{n^k}\sum_{i_1,i_2,\ldots i_k}\proj{+}_{i_1}\cdots\proj{+}_{i_k}$ into $n^k$ terms each with operator norm at most 1, we also see that $\tln{(H_0^k)} \leq 1$ for any $k$, and consequently $\tln_\eps(\proj{+}^{\otimes n}) \leq 1$ for any $\eps$.  
\end{example}

\begin{definition}[Subset operator] \label{def:subset_op}
For any $k$-local operator $R = \sum R_i$ and a subset $S \subseteq [n]$, we define the subset operator $R_S$ as
\begin{align*}
R_S = \sum_{\substack{i \sth R_i \textnormal{ acts non-trivially}\\\textnormal{on at least one qubit in $S$}}} R_i \,.
\end{align*}
\end{definition}

The decomposition of a $k$-local operator $R = \sum_i R_i$ we have considered so far is not unique.
It will occasionally be useful to define a canonical such decomposition.
This can easily be done as follows: noting that the $n$-qubit Pauli matrices (with identity) form a basis of $L(\cH)$, for any operator $R \in L(\cH)$ there is a unique decomposition in terms of the $n$-qubit Pauli matrices.
Furthermore, if $R$ is $k$-local, only basis elements that act non-trivially on at most $k$ qubits will appear in this decomposition.
Starting from this unique decomposition, we can group basis elements that act non-trivially on the same set of qubits.
This way, we obtain a unique decomposition of a $k$-local operator $R$ as
\begin{align}
R = \sum_{T \subseteq [n] \sth |T| \leq k} O_{(T)} \,, \label{eqn:unique_decomp}
\end{align}
 where $O_{(T)}$ is a (weighted) sum of all Pauli operators that act non-trivially exactly on qubits in the set $T$.
The following lemma bounds the operator norm of an individual term in this decomposition.
\begin{lemma}
\label{lem:opbound}
Given an operator $R$, consider the unique decomposition $R=\sum_T O_{(T)}$ from \cref{eqn:unique_decomp}.
Then, for all $T \subseteq [n]$:
\begin{align*}
\|O_{(T)}\| \leq 2^{|T|}\|R\| \,.
\end{align*}
\end{lemma}
\begin{proof}
Defining $O'_{(S)}\deq\sum_{T\subseteq S}O_{(T)}$, we can write $O'_{(S)}$ as a Pauli twirl applied to the operator $R$:
$$O'_{(S)} = \mathbb{E}_i \sigma^i_{[n]\setminus S}(R)\sigma^i_{[n]\setminus S}\,,$$
where $\{\sigma^i_{[n]\setminus S}\}_{i=1}^{4^{n-|S|}}$ is the set of all multi-qubit Pauli operators that act as identity on $S$. By the inclusion-exclusion principle (see e.g.~\cite[Thm. 12.1]{graham1995handbook}),
$$O_{(T)}=\sum_{U\subseteq T} (-1)^{|T\setminus U|}O'_{(U)}\,.$$
Thus,
$$\|O_{(T)}\| \leq 2^{|T|}\max_{U\subseteq T}\|O'_U\|\leq 2^{|T|}\max_{U\subseteq T}\|\mathbb{E}_i \sigma^i_{[n]\setminus U}(R)\sigma^i_{[n]\setminus U}\|= 2^{|T|} \|R\|\,,$$
where the last equality holds because the Pauli twirl is a unital channel.
\end{proof}

\subsection{Polynomial approximations}
\begin{definition}[Chebyshev polynomials] \label{def:cheby}
The Chebyshev polynomials are defined as
\begin{align*}
T_k(x)=\frac{k}{2}\sum_{r=0}^{\lfloor \frac{k}{2}\rfloor}\frac{(-1)^r}{k-r}{k-r\choose r}(2x)^{k-2r}
\deq \sum_{r=0}^{\lfloor \frac{k}{2}\rfloor} \alpha_r x^{k-2r} \,.
\end{align*}
\end{definition}

We have already seen the utility of the powering function $x \mapsto x^s$ in creating local operator approximations in \cref{ex:plusstate}.
However, the powering function does not result in approximations with an optimal tradeoff between approximation error and locality.
For this purpose, we require an approximation to the function $x \mapsto x^s$ as a linear combination of Chebyshev polynomials.

\begin{lemma} \label{lem:approx_power}
Let $T_k$ be the degree-$k$ Chebyshev polynomial and $p_{s,k}$ the probability that an $s$-step symmetric random walk on integers (starting from $0$) is at $k$ or $-k$.
For any $a \leq s$, we define the degree-$a$ polynomial 
\begin{align*}
P_{s,a}(x) = \sum_{k=0}^{a} p_{s,k} T_k(x) \,.
\end{align*}
Then, $P_{s,a}$ provides a good approximation to $x^s$ in the sense that for any operator $O$ with $\norm{O} \leq 1$, 
\begin{align*}
\norm{O^s - P_{s,a}(O)} \leq 2e^{-\frac{(a+1)^2}{2s}} \,.
\end{align*}
\end{lemma}

\begin{proof}
As shown in~\cite[Theorem 3.1]{SV14}, the monomial $x^s$ can be viewed as a random walk over the Chebyshev polynomials. More precisely, for any $s>0$ we have
\begin{align}
x^s = \sum_{k=0}^{s} p_{s,k} T_k(x) = P_{s,s}(x) \,.\label{eqn:monomial_exact}
\end{align}
By the Chernoff bound, $\sum_{k\geq a} p_{s,k} \leq 2e^{-\frac{a^2}{2s}}$. 
This means that the contributions from the high-degree Chebyshev polynomials in \cref{eqn:monomial_exact} are suppressed, and we can obtain a good approximation to $x^s$ by truncating \cref{eqn:monomial_exact} at a degree $a < s$.
Specifically, for any operator $O$ with $\norm{O} \leq 1$ we can bound 
\begin{align*}
\norm{O^s - \sum_{k=0}^{a} p_{s,k} T_k(O)} = \norm{\sum_{k=a+1}^{s} p_{s,k} T_k(O)} \leq \sum_{k=a+1}^{s} |p_{s,k}| \leq 2e^{-\frac{(a+1)^2}{2s}} \,.
\end{align*}
Here, the first inequality holds by the triangle inequality and because $|T_k(x)| \leq 1$ for $|x| \leq 1$.
\end{proof}

\subsection{Concentration bounds from local operator approximations}
\label{sec:concfromloc}

Our general strategy for proving concentration bounds on quantum states will be to show that these states can be approximated by local operators.
From this we can obtain concentration bounds from the following lemma, adapted from~\cite{KAAV15}. 
We state this lemma for the Hamming weight distribution $W_\rho$ of a state $\rho$, but it can easily be generalised to the distribution of any function of bitstrings that varies slowly as the Hamming weight is changed (see \cref{lem:qaoa_energy_conc} for an example).
\begin{lemma} \label{lem:conc_bounds_from_local}
Let $\psi= \proj{\psi}$ be a $(k, \eps)$-local pure quantum state and let $m$ be the median of $W_\psi$.
Then 
\begin{align*}
\pr{W_\psi > m + k} \leq 4 \eps^2 \quad\tand\quad \pr{W_\psi < m - k} \leq 4 \eps^2 \,.
\end{align*}

\begin{proof}
We only prove the first bound, the second is analogous.
We define 
\begin{align*}
\Pi_{> m + k} = \sum_{x \in \bits : |x| > m + k} \proj{x}
\end{align*}
as the projector onto all computational basis states with Hamming weight greater than $m+k$.
We define $\Pi_{\leq m}$ analogously.
Because strings in the support of $\Pi_{> m + k}$ and $\Pi_{\leq m}$ differ on more than $k$ positions, for any operator $O$ that only acts non-trivially on $k$ qubits we have that $\Pi_{> m + k} O \Pi_{\leq m} = 0$.
By linearity we also get $\Pi_{> m + k} R \Pi_{\leq m} = 0$ for any $k$-local operator $R$.

Because $\psi$ is $(k, \eps)$-local, there exists a $k$-local operator $R$ that is $\eps$-close to $\psi$.
We therefore get that 
\begin{align*}
\eps 
\geq \norm{\psi - R} 
\geq \norm{\Pi_{> m + k} (\psi - R) \Pi_{\leq m}}
= \norm{\Pi_{> m + k} \psi \Pi_{\leq m}} \,.
\end{align*}
By the variational definition of the operator norm and inserting $\psi = \proj{\psi}$:
\begin{align*}
\norm{\Pi_{> m + k} \psi \Pi_{\leq m}} 
&\geq \norm{\Pi_{> m + k} \proj{\psi} \Pi_{\leq m} \ket{\psi}}_2 \\
&= \bra{\psi} \Pi_{\leq m} \ket{\psi} \norm{\Pi_{> m + k} \ket{\psi}}_2 \\
&= \tr{\Pi_{\leq m} \psi} \cdot \left(\tr{\Pi_{> m + k} \psi}\right)^{1/2} \\
&= \pr{W_\psi \leq m} \cdot \left(\pr{W_\psi > m + k}\right)^{1/2} \,.
\end{align*}
The lemma now follows by noting that since $m$ is the median of $W_\psi$, $\pr{W_\psi \leq m} \geq 1/2$.
\end{proof}
\end{lemma}

More generally, we observe that approximately local quantum states have the following clustering property (see also~\cite{ANirkhe22, AB22, ABN22}).

\begin{lemma}\label{lem:local_clustering}
Let $\psi = \proj{\psi}$ be a $(k, \eps)$-local operator and consider two sets $S, S' \subset \bits^n$ with Hamming distance more than $k$ from each other.
Then, either $\tr{\Pi_{S} \rho} \leq \sqrt \eps$ or $\tr{\Pi_{S'} \rho} \leq \sqrt \eps$ (or both).
\end{lemma}
\begin{proof}
Letting $R$ be the $k$-local operator $\eps$-close to $\psi$, we can bound
\begin{align*}
\tr{\Pi_{S} \psi} \tr{\Pi_{S'} \psi} = \bra{\psi} \Pi_{S} \psi \Pi_{S'} \ket{\psi} \leq \norm{\Pi_{S} \psi \Pi_{S'}} \leq \norm{\Pi_{S} R \Pi_{S'}} + \eps = \eps \,.
\end{align*}
The last equality follows by the same reasoning as in the proof of \cref{lem:conc_bounds_from_local}.
This implies the lemma.
\end{proof}

\section{Concentration bounds for shallow circuits and matrix product states}

\cite{AAG21} showed that the ground states of various local Hamiltonians can be approximated by local operators using polynomials.
Combining this with \cref{lem:conc_bounds_from_local}, we can easily obtain concentration bounds for such states.
As concrete examples, we prove concentration bounds for the output states of shallow quantum circuits and injective matrix product states (MPSs).

We begin by considering the special case of ground states (or, for convenience, maximum-energy eigenstates) of commuting local Hamiltonians $H = \sum H_i$, where each $H_i$ is a projector.

\begin{lemma} \label{lem:ground_state_approx}
Let $H = \sum_{i = 1}^r H_i$ be a sum of $r$ commuting projectors each acting on $\ell$ qubits, and $\ket{\psi}$ the maximum-energy eigenstate of $H$.
Then, for any $d \in (\sqrt{r}, r)$, the state $\proj{\psi}$ is $\left( \ell d,  2^{-\frac{d^2}{2^8 r}} \right)$-local and 
\begin{align*}
\pr{W_\psi > m + \ell d} \leq 4\cdot 2^{-\frac{d^2}{2^7 r}} \quad\tand\quad \pr{W_\psi < m - \ell d} \leq 4\cdot 2^{-\frac{d^2}{2^7 r}} \,,
\end{align*}
where $m$ is the median of $W_\psi$.
\end{lemma}
\begin{proof}
Because $H$ is a sum of commuting projectors, $\proj{\psi} = \prod H_i$.
\cite[Theorem 3.1]{AAG21} (see also~\cite{KahnLS96, BuhrmanCWZ99}) construct multi-variate degree-$d$ polynomials $P_d$ (with $d \in (\sqrt{r}, r)$) such that for any $x_1, \dots, x_r \in \bits$:
\begin{align*}
|P_d(x_1, \dots, x_r) - \prod_i x_i| \leq 2^{-\frac{d^2}{2^8 r}} \,.
\end{align*}
If we insert the $\ell$-local projectors $x_i = H_i$, whose spectrum is $\bits$, we get that 
\begin{align*}
\norm{P_d(H_1, \cdots, H_r) - \proj{\psi}} \leq 2^{-\frac{d^2}{2^8 r}} \,,
\end{align*}
and $P_d(H_1, \cdots, H_r)$ is a $(\ell d)$-local operator by \cref{ex:poly_to_local_op}.
This shows that $\psi$ is $\left(\ell d,  2^{-\frac{d^2}{2^8 r}} \right)$-local, and the concentration bound follows directly from \cref{lem:conc_bounds_from_local}.
\end{proof}

As an application of this, we consider quantum circuits with arbitrary 2-local gates between any two qubits, arranged into $t$ layers.
The number of layers $t$ is called the circuit depth.
The key property of shallow circuits is that they spread locality in a controllable way:
by a standard lightcone argument, it is easy to see that if $O$ is a $k$-local operator and $U$ is a unitary implemented by a depth-$t$ circuit, then $U O U^\dagger$ is a $(2^t \cdot k)$-local operator.
Therefore, the output state $U \ket{0}$ of the circuit is the maximum-energy eigenstate of the $2^t$-local Hamiltonian $H = \sum_{i} U (\proj{0}_i \ot \1_{\setminus i}) U^\dagger$, and the local terms of the Hamiltonian are clearly commuting projectors.
Hence, we obtain the following corollary, answering an open question from~\cite{DMRF22}.

\begin{corollary} \label{lem:lowdepth}
Let $U$ be a unitary implemented by a depth-$t$ circuit.
Then, the output state $\proj{\psi}=U\proj{0}^{\otimes n}U^\dagger$ is $\left(k, 2^{-\frac{k^2}{2^{2t+8}\cdot n}}\right)$-local operator for $k \in (2^t\sqrt{n}, 2^tn)$.
Furthermore, denoting the median of $W_{\psi}$ by $m$, the following Gaussian concentration bounds hold:
\begin{align*}
\pr{W_\psi > m + k} \leq 4 \cdot 2^{-\frac{k^2}{2^{2t+7}\cdot n}} \quad\tand\quad \pr{W_\psi < m - k} \leq 4 \cdot 2^{-\frac{k^2}{2^{2t+7}\cdot n}}\,.
\end{align*}
\end{corollary}

A more general case considered in~\cite{AAG21} is that of a 1D local Hamiltonian $H = \sum H_i$, where $0 \leq H_i \leq \1$ and $H$ needs to satisfy the local gap property (see~\cite{AAG21} for a definition). 
For such Hamiltonians,~\cite[Theorem 3]{AAG21} gives local approximations to the ground state, and we can use this and our \cref{lem:conc_bounds_from_local} to obtain concentration bounds.
We do not spell out the full statement and instead consider a useful example, injective matrix product states (MPSs) with constant bond dimension.
We refer to~\cite{cirac2021matrix} for an introduction to MPSs.
For our purposes, the main property of injective MPSs is that they are the unique ground state of a ``parent Hamiltonian'' with constant locality spectral gap (see e.g.~\cite[Section IV.C]{cirac2021matrix} for details). In fact, the proof of the spectral gap lower bound also implies a constant local gap lower bound~\cite{PVWC07}.
From this, the following result is immediate.
\begin{lemma}
\label{lem:mps_bound}
Let $\ket{\psi}$ be an injective MPS on a chain with $n$ qubits with a constant bond dimension. Then, for every $\delta\in (0,\frac{1}{4})$,  $\proj{\psi}$ is a $\left(k, e^{-c_1(\delta) \frac{k^2}{n}} \right)$-local operator for $k\leq c_2(\delta) n^{1-\delta}$ for some $c_{1,2}(\delta)$ independent of $n$. 
Furthermore, denoting by $m$ the median of $W_{\proj{\psi}}$,
\begin{align*}
\pr{W_\rho > m + k} \leq e^{- 2 c_1(\delta) \frac{k^2}{n}} \quad\tand\quad \pr{W_\rho < m - k} \leq e^{- 2 c_1(\delta) \frac{k^2}{n}} \,.
\end{align*}
\end{lemma}

An injective MPS can be used to encode conditionally independent distributions on a line. For this, consider a distribution $P(x_1,\ldots x_n)$ such that $P(x_i| x_1,\ldots x_{i-1})=P(x_i| x_{i-1})$ and assume that all these conditional probabilities are positive.
Note that $P(x_1,\ldots x_n)=P(x_1)P(x_2|x_1)\ldots P(x_n|x_{n-1})$. 
Then, it is easy to verify that the state $\sum_{x_1,\ldots x_n}\sqrt{P(x_1,\ldots, x_n)}\ket{x_1,x_2,\ldots x_n}$ can be written as an injective MPS with constant bond dimension. The output distribution after a computational basis measurement is precisely $P$, which allows us to show Gaussian concentration using \cref{lem:mps_bound}, reproducing a (version of) Azuma's inequality.
In this sense, \cref{lem:mps_bound} can be understood as a quantum version of Azuma's inequality.

\section{Concentration bounds for dense Hamiltonian evolution} \label{sec:qaoa_conc}

In the previous section we considered the output states of shallow quantum circuits and showed that they can be approximated by local operators.
A depth-$t$ circuit can be written as a product of $t$ unitaries $U_1U_2\ldots U_t$, where each $U_i$ is a tensor product of one- or two-qubit gates acting on disjoint sets of qubits. 
Each $U_i$ can also be written as a Hamiltonian evolution $e^{\iota H^{(i)}}$, where $H^{(i)}$ is a $2$-local Hamiltonian and each qubit participates in at most one local term. 
Therefore, we can rephrase the result from the previous section as saying that for any sequence of $t$ 2-local Hamiltonians $H^{(i)}$ where each qubit participates in at most one local term, the state $e^{\iota H^{(t)}} \cdots e^{\iota H^{(1)}} \ket{0}^{\ot n}$ can be approximated by a local operator.
Since each local term in $H^{(i)}$ acts on different qubits, $H^{(i)}$ is obviously a commuting Hamiltonian.

In this section, we generalise this result to more general families of commuting Hamiltonians $H^{(i)}$.
In particular, we drop the requirement that each qubit can only participate in at most one local term of $H^{(i)}$ and instead allow local commuting Hamiltonians with dense interaction graphs, i.e.~any qubit is allowed to  participate in an arbitrary number of terms.
Evolution under such dense Hamiltonians cannot be implemented by a shallow circuit.
Nonetheless, we show that as long as such Hamiltonians satisfy a norm constraint explained in \cref{eqn:subset_condition}, the output state $e^{\iota H^{(t)}} \cdots e^{\iota H^{(1)}} \ket{0}^{\ot n}$ can still be approximated by a local operator just like for shallow circuits, and as a consequence also obeys concentration bounds.
This is of particular relevance as quantum optimisation algorithms such as the QAOA apply an evolution of this form when applied to dense constrained optimisation problems (COPs).
Therefore, our concentration bounds apply to the output of the QAOA for dense COPs, which previously required a highly technical analysis that only yielded an asymptotic statement without explicit bounds~\cite{BGMZ22}.
We can also use these concentration bounds to prove limitations on the performance of the QAOA (and dense evolutions more generally) at solving COPs; see \cref{sec:limitations} for details.

More formally, let $H^{(1)},H^{(2)},\ldots H^{(p)}$ be a collection of $\ell$-local commuting Hamiltonians (i.e.~the local terms in each Hamiltonian commute, but the different $H^{(i)}$ need not commute), where each qubit is allowed to participate in arbitrarily many local terms of each Hamiltonian $H^{(i)}$. Define $U_{i}=e^{-\iota H^{(i)}}$.
Because each Hamiltonian $H^{(i)}$ may have a dense interaction graph, we call $U_{1}\cdots U_p$ a \emph{dense  Hamiltonian evolution} with \emph{level} $p$.
To prove concentration bounds, we will require that there exist constants $\alpha \in [0, 1)$ and $\tilde C > 1$ such that for every $i$ and every subset $S \subseteq [n]$ of qubits, the subset Hamiltonian $H^{(i)}_S$ (see \cref{def:subset_op}) satisfies 
\begin{align}
\norm{H^{(i)}_S} \leq \tilde C n^\alpha |S|^{1 - \alpha} \,. \label{eqn:subset_condition}
\end{align}
In the special case of sparse Hamiltonians, we get that \cref{eqn:subset_condition} is satisfied with $\alpha = 0$. However, crucially, for random \emph{dense} classical COPs, \cref{eqn:subset_condition} can still be satisfied with high probability.
For example, in \cref{sec:sk_model} we show that a class of random spin models satisfies \cref{eqn:subset_condition} with high probability even though it allows for constraints between any variables, i.e.~its constraint (hyper-)graph is the complete (hyper-)graph.
The well-known Sherrington-Kirkpatrick model is an example of such a spin model.

We consider any pure product state $\rho_0$ and denote by
\begin{align}
\rho_p = \Big( U_p \cdots U_1\Big) \rho_0  \Big( U_1^\dagger  \cdots U^{\dagger}_p\Big)  \label{eqn:dense_state}
\end{align}
the output of the dense Hamiltonian evolution. 
The purpose of this section is to show that the output state $\rho_p$ satisfies certain concentration properties even if $p$ grows (slowly) with $n$.
Our strategy for proving such concentration bounds will be to approximate the state $\rho_p$ by a local operator (\cref{thm:qaoamain}).
Once we have established such a local approximation, \cref{lem:conc_bounds_from_local} immediately implies a concentration bound for the Hamming weight distribution $W_{\rho_p}$ (\cref{lem:qaoa_weight_conc}).
In~\cref{lem:qaoa_energy_conc}, we extend this to a concentration bound for the energy density of $\rho_p$ with respect to any classical Hamiltonian satisfying \cref{eqn:norm_bound}.

\subsection{Local approximations to output states of dense Hamiltonian evolution}
We start by giving the main result of this section, a local approximation to the output state $\rho_p$.
The proof of this result proceeds inductively: we first show how to approximate the starting state $\rho_0$ by a local operator; then, we can analyse how the locality and approximation error evolves under application of the unitaries $U_i$.
We emphasise that the bounds in \cref{thm:qaoamain} are optimised for ease of use, not tightness; one can easily obtain a tighter final result by keeping around more parameters.

\begin{theorem}
\label{thm:qaoamain}
Consider the output state $\rho_p$ from \cref{eqn:dense_state} for a family of $\ell$-local commuting Hamiltonians $H^{(1)},\ldots H^{(p)}$ satisfying \cref{eqn:subset_condition} for some $\alpha \in [0, 1)$ and $\tilde C > 0$ (independent of $n$). 
Then, for sufficiently large $n$ and $p = o(\log(n))$, there exists a constant $c_1 > 0$ such that $\rho_p$ is $(k_p,\eps_p)$-local with
\begin{align*}
k_p \leq c_1^p n^{1 - (1-\alpha)^p/4}\,,\quad
\eps_p \leq 4 e^{-n^{1/8}/\sqrt{2}} \,.
\end{align*}
\end{theorem}

\begin{remark}
In general, the above bound is useful when $k = o(n)$.
For sparse COPs, where $\alpha = 0$, the bound on $k_p$ simplifies to $k_p \leq c_1^p n^{3/4}$.
Therefore, we get $k_p = o(n)$ for $p = o(\log n)$.
(Note that for sparse Hamiltonians, the circuit is in fact low depth, so one can alternatively use \cref{lem:lowdepth}.)
For dense COPs, where $\alpha > 0$, we need $p = o(\log \log n)$ for $k_p = o(n)$.
For the remainder of this section, we will focus on dense Hamiltonians and always impose the requirement that $p = o(\log\log n)$.
\end{remark}

\begin{remark}
As an example of why such a local operator approximation is useful, we observe that \cref{thm:qaoamain} combined with \cref{lem:local_clustering} gives a clustering property for the output of dense evolutions.
Such a clustering property is used in the proof of the NLTS theorem~\cite{ABN22}, which shows that the low energy states of recently discovered~\cite{panteleevk2021, quantum-tanner-codes} linear-rank and linear-distance quantum LDPC code Hamiltonians require $\Omega(\log n)$ quantum circuit depth to prepare. 
Replacing~\cite[Fact 4]{ABN22} with the clustering property for the output states of dense evolutions, we can show that if one were to try to prepare the low-energy states of LDPC codes using dense Hamiltonian evolution instead of shallow circuits, one would need at least $\Omega(\log\log n)$ levels of dense evolution.
\end{remark}

\begin{proof}[Proof of \cref{thm:qaoamain}]
We invoke \cref{lem:qaoa_local_approx} shown below, which gives general (albeit complicated) expressions for $k_p$ and $\eps_p$. Setting $c_1 = 40 \cdot \ell \tilde C$, which is a constant by assumption, we immediately obtain the bound on $k_p$.
To bound $\eps_p$, we use the expression from \cref{lem:qaoa_local_approx}:
\begin{align*}
\eps_p  \leq 3 \, e^{-n^{1/8}/\sqrt{2}} + 6 \sum_{j = 1}^p e^{ - 4  \cdot (20 \cdot l)^{j-1} \tilde C^j  n^{1 - (1-\alpha)^j/4} + \sqrt{n} + 2 (j - 1) \log(2 n) k_{j-1}} \,.
\end{align*}
We first bound the exponent of the second term for $j \in \{1, \dots, p\}$:
\begin{align*}
&- 4\cdot (20 \cdot l)^{j-1} \tilde C^j n^{1 - (1-\alpha)^j/4} + \sqrt{n} + 2 (j - 1) \log(2 n) k_{j-1}  \\
&\leq - \tilde C n^{1 - (1-\alpha)^j/4} + \sqrt{n} + 2 (p - 1) \log(2 n) c_1^p n^{1 - (1-\alpha)^{j-1}/4} 
\end{align*}
For $p = o(\log n)$ and sufficiently large $n$ and $p$, we see that the first term $- \tilde C n^{1 - (1-\alpha)^j/4}$ is dominant for any $j$.
Therefore, there exist constants $c_3, c_4$ such that
\begin{align*}
\eps_p \leq 3 e^{-n^{1/8}/\sqrt{2}} + 6 p c_3 e^{-c_4 n^{3/4}} \leq 4 e^{-n^{1/8}/\sqrt{2}} 
\end{align*}
for sufficiently large $n$ and $p = o(\log(n))$.
\end{proof}

As mentioned above, \cref{thm:qaoamain} is a simplification of the following technical lemma. 

\begin{lemma} \label{lem:qaoa_local_approx}
Consider the output state $\rho_p$ from \cref{eqn:dense_state} for a family of $\ell$-local commuting Hamiltonians $H^{(1)},\ldots H^{(p)}$ satisfying \cref{eqn:subset_condition} for some $\alpha \in [0, 1)$ and $\tilde C$ such that  $n\geq \frac{k}{\tilde C^{1/\alpha}}$.
Then, $\rho_p$ is $(k_p,\eps_p)$-local for 
\begin{align*}
k_p \leq 2 \cdot (20 \cdot \ell \tilde C)^{p} n^{1 - (1 - \alpha)^p /4} \,,\quad 
\eps_p  &\leq 3 \, e^{-n^{1/8}/\sqrt{2}} + 6 \sum_{j = 1}^p e^{ - 4  \cdot (20 \cdot l)^{j-1} \tilde C^j  n^{1 - (1-\alpha)^j/4} + \sqrt{n} + 2 (j - 1) \log(2 n) k_{j-1}} \,.
\end{align*}
\end{lemma}

\begin{proof}
For $i \leq p$ we define the intermediate states $\rho_i$ in the obvious way, i.e.
\begin{align*}
\rho_i = \Big( U_i \cdots U_1\Big) \rho_0  \Big( U_1^\dagger  \cdots U_i^{\dagger}\Big) \,.
\end{align*}
We now claim that each $\rho_i$ is $(k_i, \eps_i)$-local with
\begin{align*}
k_i &\leq 2 \cdot (20 \cdot \ell \tilde C)^{i} n^{1 - (1 - \alpha)^i /4} \,,\\
\eps_{i} &\leq 3\, e^{-n^{1/8}/\sqrt{2}} + 6 \sum_{j = 1}^i e^{ - 4 \cdot (20 \cdot l)^{j-1} \tilde C^j  n^{1 - (1-\alpha)^j/4} + \sqrt{n} + 2 (j - 1) \log(2 n) k_{j-1}}  \,, \\
\tln_{\eps_i}(\rho_i) &\leq 2\, e^{\sqrt{n} + 2 i \log(2n) k_i} \,. \numberthis \label{eqn:inductive_claim}
\end{align*}
We prove this by induction.
The base case $i = 0$, i.e.~the approximation of the starting state $\rho_0$, follows from \cref{lem:starting_state_approx}.
For the inductive step, we will make use of a simplification of \cref{lem:comm_locality_spread}, stated as \cref{lem:qaoa_spread} below.
Concretely, suppose that \cref{eqn:inductive_claim} holds for $\rho_i$.
Since 
\begin{align*}
\rho_{i+1} = U_{i+1} \rho_i U_{i+1}^\dagger \,,
\end{align*}
we can apply \cref{lem:qaoa_spread} with 
\begin{align*}
\delta = (1 - \alpha)^i / 4 \quad \tand \quad C = 2\cdot (20\cdot l \tilde C)^{i} \,.
\end{align*}
This implies the bounds in \cref{eqn:inductive_claim} after minor simplifications.
\end{proof}

We now show the two missing statements in the proof of \cref{lem:qaoa_local_approx}: \cref{lem:starting_state_approx} for the base case of the induction and \cref{lem:qaoa_spread} for the inductive step.
We will show both of these statements in slightly more generality than is required for \cref{lem:qaoa_local_approx}, with additional parameters that could be optimised to obtain tighter bounds in \cref{lem:qaoa_local_approx} for specific applications.

For the base case of the induction, we need to approximate the starting state $\rho_0$.
By assumption, this is a pure product state.
Since all pure product states are related to each other by 1-local unitaries, it suffices to show the lemma for any one particular product state.
For simplicity, we show it for $\rho_0 = \proj{+}^{\ot n}$, which we have already considered in \cref{ex:plusstate}.
There, we gave the following simple approximation:
we defined the local Hamiltonian $H_0 = \frac{1}{n} \sum_{i = 1}^n \proj{+}_i$ and noted that $\proj{+}^{\ot n}$ can be approximated by powers of this Hamiltonian.
Here, we will require a better tradeoff between locality and approximation error.
This will be achieved using the approximation to the powering function established in \cref{lem:approx_power}.
Because we additionally need to bound the total local norm of the approximation, which becomes large when using the function from \cref{lem:approx_power}, we will combine the simple powering from \cref{ex:plusstate} with the function from \cref{lem:approx_power} to achieve an approximation that has both low locality and low total local norm.
The requirement that our approximation should have small total local norm is also the reason why we cannot use \cref{lem:ground_state_approx} to approximate $\rho_0$ in this setting.

\begin{lemma} \label{lem:starting_state_approx}
The state $\rho_0 \deq \proj{+}^{\ot n}$ is $(k_0, {\eps_0})$-local with 
\begin{align*}
k_0 = 2 n^{3/4} \,, \quad 
{\eps_0} = 3e^{-(n^{1/8})/\sqrt 2}\,,\quad 
\tln_{\eps_0}(\rho_0) \leq 2 e^{\sqrt n}.
\end{align*}
\end{lemma}

\begin{proof}
Let $H_0 = \frac{1}{n} \sum_{i = 1}^n \proj{+}_i$ and consider the polynomial $P_{s,a}(x) = \sum_{j=0}^{a} p_{s,j} T_j(x)$ from \cref{lem:approx_power}.
We define the $(a \cdot m)$-local operator 
\begin{align*}
R = P_{s,a}(H_0^m) \,,
\end{align*}
where $a, m$, and $s$ will be chosen later.
We can then bound
$$\|\rho_0 - R\|\leq \|\rho_0 - (H_0^m)^s\|+\|(H_0^m)^s - P_{s,a}(H_0^m)\|\leq e^{-\frac{ms}{n}} + 2e^{-\frac{(a+1)^2}{2s}}\,,$$
where the second inequality uses \cref{ex:plusstate} and \cref{lem:approx_power}.
If we now choose $m=\lceil n^{1/4} \rceil$, $a=\lceil n^{1/2} \rceil$, and $s=\lceil \frac{1}{\sqrt{2}} n^{7/8}\rceil$, then we get that 
\begin{align*}
\frac{m s}{n} \geq \frac{1}{\sqrt{2}}n^{1/8} \quad\tand\quad 
\frac{(a+1)^2}{2s} = \frac{(\lceil n^{1/2} \rceil+1)^2}{2\lceil \frac{1}{\sqrt{2}} n^{7/8}\rceil} 
\geq n^{1/2} \frac{n^{1/2} + 1}{ \sqrt{2} n^{7/8} + 1} \geq \frac{1}{\sqrt{2}}n^{1/8} \,,
\end{align*}
where the last inequality holds because $\sqrt{2} n^{7/8} \geq n^{1/2}$.
Consequently, $R$ has locality $k_0 = a\cdot m \leq 2 n^{3/4}$ and approximates $\rho_0$ up to error ${\eps_0} \leq e^{-\frac{ms}{n}} + 2e^{-\frac{(a+1)^2}{2s}} \leq 3e^{-(n^{1/8})/\sqrt 2}$ as claimed.

To bound $\tln_{{\eps_0}}(\rho_0) = \tln(R)$, we can first bound the norms of the coefficients in the Chebyshev polynomial (\cref{def:cheby}) by
$$\sum_{r=0}^{\lfloor \frac{j}{2}\rfloor} |\alpha_r|= \sum_{r=0}^{\lfloor \frac{j}{2}\rfloor}\frac{j}{2}\cdot\frac{1}{j-r}{j-r\choose r}2^{j-2r}\leq \sum_{r=0}^{\lfloor \frac{j}{2}\rfloor}{j-r\choose r}2^{j-2r}\leq 2^j \sum_{r=0}^{\lfloor \frac{j}{2}\rfloor}{j\choose r}\frac{1}{4^r}\leq 2^j\cdot \frac{5^j}{4^j}\leq e^j.$$
The second-to-last inequality holds because $\sum_{r=0}^{\lfloor \frac{j}{2}\rfloor}{j\choose r}\frac{1}{4^r} \leq \sum_{r=0}^{j}{j\choose r}\frac{1}{4^r} = (1 + \frac{1}{4})^j$ by the binomial formula.
Therefore,
$$\tln(R) = \tln\left(\sum_{j=0}^{a} p_{s,j} T_j(H_0^m)\right)\leq \sum_{j=0}^{a}\sum_{r=0}^{\lfloor \frac{j}{2}\rfloor} p_{s,j}\cdot |\alpha_r|\cdot \tln\left(H_0^{m(j-2r)}\right) \leq \sum_{j=0}^{a} p_{s,j}\cdot e^j\leq e^a.$$
The second inequality holds because $\tln\left(H_0^{m(j-2r)}\right) = 1$ as noted in \cref{ex:plusstate}, and the last inequality holds because $\sum_{j=0}^{a} p_{s,j} \leq 1$ as $p_{s, j}$ forms a probability distribution.
Inserting $a = \lceil n^{1/2} \rceil \leq n^{1/2} + \frac{1}{2}$, we get $e^a \leq 2 e^{n^{1/2}}$.
\end{proof}

For the inductive step in \cref{lem:qaoa_local_approx}, we need to analyse how the locality and approximation error change under application of one unitary $U_i$.
This is done in the following lemma, which is a simplification of \cref{lem:comm_locality_spread}.
If one wishes to obtain the tightest possible bounds in \cref{lem:qaoa_local_approx}, one can use \cref{lem:comm_locality_spread} directly instead of this simplified statement.

\begin{lemma} \label{lem:qaoa_spread}
Let $\rho$ be a $(k, \eps)$-local quantum state with $k \leq C n^{1 - \delta}$ for some $C > 1, \delta > 0$, and $H$ an $\ell$-local commuting Hamiltonian satisfying \cref{eqn:subset_condition} for some $\alpha \in [0, 1)$ and $n \geq \frac{k}{\tilde C^{1/\alpha}}$.
Then, the state 
\begin{align*}
\rho' \deq e^{-\iota H} \rho e^{\iota H}
\end{align*}
is $(k',\eps')$-local for 
\begin{align*}
k' \leq 20 \ell \tilde C C n^{1 - (1 - \alpha) \delta} \,, \quad 
\eps' = 3 e^{-4 \tilde C C n^{1 - (1 - \alpha) \delta}} \tln_{\eps}(\rho) + \eps \,, \quad 
\tln_{\eps'}(\rho') \leq e^{2\log(2n) k'} \tln_{\eps}(\rho) .
\end{align*}
\end{lemma}
\begin{proof}
Let $R$ be a $k$-local operator within $\eps$-distance from $\rho$.
We apply \cref{lem:comm_locality_spread} to $R$ for the Hamiltonian $H$, which satisfies $\norm{H_S} \leq \tilde C n^\alpha |S|^{1 - \alpha}$.
Setting $\mu=1+4/e$, we get that $e^{-\iota H} R e^{\iota H}$ is a $(k', \tilde\eps)$-local operator with
\begin{align*}
k' \leq 2\ell \lceil 7\tilde C n^\alpha \kappa^{1 - \alpha} \rceil + k \,, \quad 
\tilde\eps = 3 e^{-4\,\tilde C n^\alpha \kappa^{1 - \alpha}} \tln(R) \,, \quad 
\tln_{\tilde\eps}\left(e^{-\iota H} R e^{\iota H} \right) \leq (2 n)^{k'} \left( \tln(R) + \tilde \eps \right) \,,
\end{align*}
for $\kappa \geq k$ to be chosen later.
Since $e^{-\iota H} R e^{\iota H}$ is $(k', \tilde\eps)$-local, it immediately follows from the triangle inequality that $e^{-\iota H} \rho e^{\iota H}$ is $(k', \eps')$-local for $\eps' = \tilde\eps + \eps$.
We can now simplify the resulting bounds as follows choosing $\kappa = C^{\frac{1}{1-\alpha}} n^{1-\delta} \geq k$:
\begin{enumerate}
\item Because $l\geq 1$, $n \geq \frac{k}{\tilde C^{1/\alpha}}$, and $\kappa \geq k$, it is clear that
$
2 \ell \lceil 7\tilde C n^\alpha \kappa^{1 - \alpha} \rceil + k \leq 20 \ell\tilde C n^\alpha \kappa^{1 - \alpha}$.
Inserting $\kappa = C^{\frac{1}{1-\alpha}} n^{1-\delta}$, we get the claimed bound on $k'$.
\item The bound on $\eps'$ follows immediately from $\eps' = \tilde\eps + \eps$, $\tln(R) = \tln_{\eps}(\rho)$, and $\kappa = C^{\frac{1}{1-\alpha}} n^{1-\delta}$.
\item To bound $\tln_{\eps'}(\rho')$, observe that since $\rho$ is a quantum state, we are only interested in approximations where $\tilde \eps \leq 1 \leq (2n)^{k'} \tln(R)$, as otherwise the claim becomes trivial.
Combining this and inserting $\tln(R) = \tln_{\eps}(\rho)$ yields the claimed bound. \qedhere
\end{enumerate}
\end{proof}

\subsection{Concentration bounds for output states of dense Hamiltonian evolution}

Combining \cref{lem:conc_bounds_from_local} and \cref{thm:qaoamain}, we immediately obtain the following concentration bound.
One can of course also derive an analogous but tighter and more explicit statement from \cref{lem:qaoa_local_approx} instead of \cref{thm:qaoamain}.
\begin{corollary} \label{lem:qaoa_weight_conc}
Let $\rho_p$ and $c_1$ be as in \cref{thm:qaoamain}, $p=o(\log\log n)$, and $m$ be the median of $W_{\rho_p}$. Then for sufficiently large $n$ and $p$,
\begin{align*}
\pr{|W_{\rho_p}-m| > c_1^pn^{1-(1-\alpha)^p/4}} \leq 128 \, e^{-\sqrt{2} n^{1/8}}.
\end{align*}
\end{corollary}

While the above statement is about concentration with respect to Hamming weight, we can also prove concentration with respect to other observables. Let $G=\sum_i G_i$ be a classical local Hamiltonian (i.e.~a local Hamiltonian that is diagonal in the computational basis) that satisfies the following condition analogous to \cref{eqn:subset_condition}: for all $S \subset [n]$, 
$$\|G_S\|\leq D n^{\alpha'}|S|^{1-\alpha'}.$$
In this case, we can also prove a concentration bound on the expectation of $\rho_p$ with respect to $G$.
More specifically, we can define a random variable $E_{\rho_p}$ indicating the ``energy'' of $\rho_p$ according to $G$, i.e.~if we take the spectral decomposition $G = \sum g_i \Pi_i$ for orthogonal projectors $\Pi_i$, then $E_{\rho_p}$ takes value $g_i$ with probability $\tr{\Pi_i \rho_p}$.

\begin{corollary} \label{lem:qaoa_energy_conc}
Let $\rho_p$ and $c_1$ be as in \cref{thm:qaoamain} with $p=o(\log\log n)$, and $G$ and $E_{\rho_p}$ as above. Let $E$ be the median energy, i.e.~the median of $E_{\rho_p}$. Then
\begin{align*}
\pr{|E_{\rho_p}-E| > 2D c_1^{p(1-\alpha')}n^{1-(1-\alpha)^{p}(1-\alpha')/4}} \leq 128 \, e^{-\sqrt{2} n^{1/8}}\,.
\end{align*}

\end{corollary}
\begin{proof}
Consider strings $x,y \in \bits^n$ and let $S\deq\{i:x_i\neq y_i\}$. Suppose $|S|\leq c_1^pn^{1-(1-\alpha)^p/4}$.  Expanding $G=\sum_{i} G_i$, consider the energy difference 
\begin{align*}
\tr{(\proj{x}-\proj{y})G} &= \sum_{i} \tr{(\proj{x}-\proj{y})G_i} \\
&= \sum_{\substack{i \sth G_i \textnormal{ acts non-trivially}\\\textnormal{on at least one qubit in $S$}}} \tr{(\proj{x}-\proj{y})G_i}\\
&= \tr{(\proj{x}-\proj{y})G_S}\\
&\leq 2\|G_S\|\leq 2D n^{\alpha'}|S|^{1-\alpha'}\leq 2D c_1^{p(1-\alpha')}n^{1-(1-\alpha)^{p}(1-\alpha')/4}\,.
\end{align*}
Let 
\begin{align*}
\Pi_{>E+f}=\sum_{x\in \bits^n: \bra{x}H\ket{x}> E+f}\proj{x},\quad \Pi_{\leq E}=\sum_{x\in \bits^n: \bra{x}H\ket{x}\leq E}\proj{x}\,.
\end{align*}
If $f> 2D c_1^{p(1-\alpha')}n^{1-(1-\alpha)^{p}(1-\alpha')/4}$, then by the above argument strings in the support of $\Pi_{>E+f}$ and $\Pi_{\leq E}$ differ on more than $c_1^pn^{1-(1-\alpha)^p/4}$ positions.
Thus, a $(c_1^pn^{1-(1-\alpha)^p/4})$-local operator $O$ satisfies $\Pi_{>E+f}O \Pi_{\leq E}=0$. 
The corollary now follows along the same lines as \cref{lem:conc_bounds_from_local}.   
\end{proof}

\subsection{Example: random spin models} \label{sec:sk_model}
As an example of a family of dense Hamiltonians that satisfies \cref{eqn:subset_condition}, we consider the pure $q$-spin model, which was also considered in~\cite{BGMZ22}.
The pure $q$-spin model is a random COP with cost function
\begin{align}
C_n^q(z; J) = \frac{1}{n^{(q+1)/2}}\sum_{i_1,\ldots i_q=1}^n J_{i_1,\ldots i_q} z_{i_1}\ldots z_{i_q} \,, \label{eqn:classical_spin_model}
\end{align}
where the coefficients $J = (J_{i_1,\ldots i_q})_{i_1,\ldots i_q \in [n]}$ are i.i.d.~standard Gaussian random variables $J_{i_1,\ldots i_q} \sim \mathcal{N}(0,1)$.
Here, $z_i \in \{\pm 1\}$ and the objective is to maximise $C^q(z_1, \dots, z_n)$.
This can be identified with a $q$-local Hamiltonian 
\begin{align}
H_n^q(J) = \frac{1}{n^{(q-1)/2}} \sum_{i_1,\ldots i_q=1}^n J_{i_1,\ldots i_q} \sigma^Z_{i_1}\ldots \sigma^Z_{i_q} \,. \label{eqn:hamiltonian_spin_model}
\end{align}
We note the different normalisation factors: $C_n^q(z)$ is normalised such that on average over $J$, $\max _{z} C_n^q(z ; J) = \Theta(1)$.
In contrast, $H_n$ has an additional factor of $n$, so that on average over $J$, $\norm{H_n^q(J)} = \Theta(n)$.
We use these different normalisations because the former is common in the classical literature (see e.g.~\cite{gamarnik2020low}), whereas the latter is common in the quantum literature.

The following lemma shows that this model satisfies \cref{eqn:subset_condition} with overwhelming probability, and as a result we can apply \cref{lem:qaoa_weight_conc} and \cref{lem:qaoa_energy_conc} to obtain concentration bounds.

\begin{lemma} \label{lem:pure_spin}
With probability at least $1-e^{-n}$ over the choice of $J_{i_1,\ldots i_q}\sim \mathcal{N}(0,1)$, the Hamiltonian $H_n^q(J) = \frac{1}{n^{(q-1)/2}} \sum_{i_1,\ldots i_q=1}^n J_{i_1,\ldots i_q} \sigma^Z_{i_1}\ldots \sigma^Z_{i_q}$ satisfies \cref{eqn:subset_condition} with $\alpha=\frac{1}{2}$ and $\tilde C=\sqrt{6}$ for every subset $S \subseteq [n]$.
\end{lemma}
\begin{proof}
Fix any subset $S \subseteq [n]$.
Recall the definition of the subset Hamiltonian from \cref{def:subset_op}, and define analogously 
\begin{align*}
C_{n, S}^q(z; J) = \frac{1}{n^{(q+1)/2}} \sum_{\{i_1,\ldots i_q\}\cap S \neq \emptyset}^n J_{i_1,\ldots i_q} z_{i_1}\dots z_{i_q} \,.
\end{align*}
Because all terms of $H_{n, S}^q$ are proportional to tensor products of Pauli-Z operators, it is easy to see that $\norm{H_{n, S}^q} = n \cdot \max_{z_1, \dots, z_n \in \{\pm 1\}} C_{n, S}(z; J)$.
For any fixed choice of $z_1, \dots, z_n \in \{\pm 1\}$, the random variable $C_{n, S}(z; J)$ is a sum of $\ell \leq |S| {n \choose q-1} \leq |S| n^{q-1}$ standard Gaussians with a normalisation factor $\frac{1}{n^{(q+1)/2}}$, and is therefore distributed as $\mathcal{N}(0,\ell/n^{q+1})$.
By the standard upper deviation inequality for Gaussians, we have that 
\begin{align*}
\prs{J_{i_1,\ldots i_q} \sim \mathcal{N}(0,1)}{ C_{n, S}^q(z; J) \geq \sqrt{6 |S| / n}} \leq e^{-3 n} \,.
\end{align*}
Since we are interested in upper-bounding the probability that $C^q_S(z; J) \geq \sqrt{6 |S|/n}$ \emph{for any} choices of $z_1, \dots, z_n$ and $|S|$, we can apply the union bound over the possible $2^n \cdot 2^n \leq e^{2n}$ choices of $z_1, \dots, z_n$ and $|S|$.
We therefore see that the probability that \cref{eqn:subset_condition} is violated (for any $S \subseteq [n]$) is at most $e^{-3n}\cdot e^{2n} = e^{-n}$ as claimed.
\end{proof}

\cite{BGMZ22} also consider a mixed $q$-spin model, which is a sum over pure $j$-spin models for $j \leq q$. Specifically, the cost function can be written as 
\begin{align*}
C^{q, \setft{mixed}}_n(z; J) = \sum_{j = 1}^q c_j \, C^j_n(z; J) \,,
\end{align*}
where $c_j$ are arbitrary real coefficients and $C^j_n$ is as defined in \cref{eqn:classical_spin_model}.
We can again  associate a Hamiltonian $H^{q, \setft{mixed}}_n$ with this cost function.
The following corollary follows immediately from \cref{lem:pure_spin} by the triangle inequality.
\begin{corollary} \label{lem:mixed_spin}
With probability at least $1 - e^{-n}$, the Hamiltonian $H^{q, \setft{mixed}}_n$ satisfies \cref{eqn:subset_condition} for $\alpha = \frac{1}{2}$ and $\tilde C = \sqrt{6} \sum |c_j|$.
\end{corollary}
We can use this property of the (mixed) random spin model to obtain concentration bounds for the output states of the QAOA applied to the COPs $C^{q}_n(z; J)$ and $C^{q, \setft{mixed}}_n(z; J)$.
This in turn can also be used to prove limitations on the success probability of the QAOA on these COPs.
We spell this out in detail in \cref{cor:qaoa_spin_model}.

\section{Limitations on dense evolutions for constraint optimisation problems} \label{sec:limitations}

Using our local operator approximations and concentration bounds for the output states of dense evolutions, we can show that such states have limitations as optimisers for COPs.
We begin by introducing a structural property of random COPs, called the overlap gap property (OGP), that roughly says that good solutions to a COP must cluster, i.e.~different good solutions must either be close to each other or far from each other.
We then combine the OGP with our concentration results and the symmetry of the QAOA output to show that for most instances of random spin models, the QAOA can only produce a good solution with negligible probability.

\subsection{Overlap gap property and existence of high-weight sets for local quantum optimisers}

We consider an objective function $C_n(z)$ for $z = (z_1, \dots, z_n) \in \pmset^n$ that we want to maximise.
We begin by recalling a few general definitions from~\cite{gamarnik2020low}, adapted to the case where algorithms for COPs output probability distributions or quantum states rather than a single element of $\pmset^n$.
\begin{definition}
For parameters $\mu \in \R$ and $\delta \in [0,1]$, we say that a probability distribution $P$ over $\pmset^n$ $(\mu,\delta)$-optimises the objective $C_n(z)$ if
\begin{align*}
\prs{z \sim P}{C_n(z) \geq \mu} \geq 1 - \delta \,.
\end{align*}
We will use the same notation for quantum states $\rho$, which we identify with a probability distribution over $\pmset^n$ in the natural way, i.e.~$\prs{\rho}{(-1)^{x}} = \bra{x} \rho \ket{x}$ for $x \in \bits^n$.
\end{definition}

\begin{definition}\label{def:OGP}
An objective function $C_n(z)$ satisfies the $(\mu,\nu_1,\nu_2)$\emph{-overlap gap property (OGP)} with parameters $\mu \in \R$ and $0\leq\nu_1<\nu_2\leq 1$ if for every $z_1,z_2\in\pmset^n$ satisfying
$C_n(z_1) \geq \mu$ and $C_n(z_2) \geq \mu$, we have that
$\frac{1}{n} |\langle z_1, z_2 \rangle| \in [0,\nu_1]\cup [\nu_2,1]$.
Here, $\langle z_1, z_2 \rangle$ denotes the usual inner product of vectors.
\end{definition}

Suppose that the objective function $C_n(z)$ has the $(\mu,\nu_1,\nu_2)$-OGP.
Then, we can define the set of ``good outputs''
\begin{align*}
G_n = \{x \in \bits^n \;|\; C_n((-1)^{x}) \geq \mu \} \,.
\end{align*}
Because $C_n(z)$ has the $(\mu,\nu_1,\nu_2)$-OGP, any two $x, x' \in G_n$ must satisfy 
$
\frac{1}{n} |\langle (-1)^x, (-1)^{x'} \rangle| \in [0,\nu_1]\cup [\nu_2,1] \,.
$
Since the Hamming distance between $x$ and $x'$ is given by $|x - x'| = \frac{n - \langle (-1)^x, (-1)^{x'} \rangle}{2}$, this implies that
\begin{align*}
|x - x'| \in \left[ 0, \tilde \nu_1 \cdot n \right] \cup \left[ \tilde \nu_2 \cdot n , n\right]
\end{align*}
for $\tilde \nu_1 \deq \frac{1 - \nu_2}{2} < \tilde \nu_2 \deq \frac{1 - \nu_1}{2}$.
Assuming that $2 \tilde \nu_1 < \tilde \nu_2$, we can partition $G_n = \cup_i S^{i}_n$ into sets (or clusters) $S^i_n$ such that for all $i \neq j$:
\begin{align}
x, x' \in S^{i}_n \implies |x - x'| \leq \tilde \nu_1 \cdot n \quad\tand\quad x \in S^{i}_n, x' \in S^{j}_n \implies |x - x'| \geq  \tilde \nu_2 \cdot n\,. \label{eqn:cluster_condition}
\end{align}
We note that the condition $2 \tilde \nu_1 < \tilde \nu_2$ is necessary for this clustering property to hold.
To see this intuitively, consider three points $x_1, x_2, x_3 \in G_n$.
The $(\mu,\nu_1,\nu_2)$-OGP then requires that $|x_i - x_j| \in \left[ 0, \tilde \nu_1 n \right] \cup \left[ \tilde \nu_2 n , n\right]$ for all pairs $(i,j)$.
Without any condition on $\tilde \nu_1$ and $\tilde \nu_2$, this would allow the following situation: the three points could be arranged ``on a line'' in the sense that $|x_1 - x_2| \leq \tilde \nu_1 n$ and $|x_2 - x_3| \leq \tilde \nu_1 n$, but $|x_1 - x_3| \geq \tilde \nu_2 n$.
This means that the points are not clustered.
However, if $2 \tilde \nu_1 < \tilde \nu_2$, then by the triangle inequality $|x_1 - x_2| \leq \tilde \nu_1 n$ and $|x_2 - x_3| \leq \tilde \nu_1 n$ together imply $|x_1 - x_3| \leq \tilde 2 \nu_1 n < \tilde \nu_2 n$.
By the OGP this means that we must in fact have $|x_1 - x_3| \leq \tilde \nu_1 n$, so we get the clustering behaviour described above.

We can now show that for (approximately) local quantum states that optimise $C_n(z)$, the measurement distribution must be concentrated on one of these sets, which we will call the \emph{high-weight set}.
\begin{lemma} \label{lem:unique_cluster}
Suppose the objective $C_n(z)$ has the $(\mu, \nu_1, \nu_2)$-OGP with $2 \tilde \nu_1 < \tilde \nu_2$ and there exists a $(k, \eps)$-local operator $\psi_n = \proj{\psi_n}$ for $k = o(n)$ that $(\mu, 1 - 4\sqrt \eps)$-optimises $C_n(z)$.
Then there exists a (unique) $i$ such that for sufficiently large $n$, $\tr{\Pi_{S^{i}_n} \psi_n} \geq \tr{\Pi_{G_n}\psi_n} - \sqrt \eps$.
\end{lemma}
\begin{proof}
Because $k = o(n)$, for sufficiently large $n$, we have $k < \tilde \nu_2 n$.
Therefore, by \cref{lem:local_clustering}, it suffices to show that there exists an $S^{i}_n$ for which $\tr{\Pi_{S^{i}_n} \psi_n} > \sqrt{\eps}$.
This is because if such an $S^{i}_n$ exists, we can consider the set $S' = G_n \setminus S^i_n$, which is at least $\tilde \nu_2 n$-far from $S^{i}_n$.
By \cref{lem:local_clustering} and the assumption $\tr{\Pi_{S^{i}_n} \psi_n} > \sqrt{\eps}$, this means that $\tr{\Pi_{S'} \psi_n} \leq \sqrt{\eps}$.
Therefore, $\tr{\Pi_{S^{i}_n} \psi_n} = \tr{G_n \psi_n} - \tr{S' \psi_n} \geq \tr{G_n \psi_n} - \sqrt \eps$.

Now suppose for the sake of contradiction that for all $i$, $\tr{\Pi_{S^{i}_n} \psi_n} \leq \sqrt{\eps}$.
Since $\psi_n$ is assumed to be a $(\mu, 1 - 4\sqrt \eps)$-optimiser, $\sum_i \tr{\Pi_{S^{i}_n} \psi_n} = \tr{G_n \psi_n} \geq 4 \sqrt \eps$.
Therefore, we can find two disjoint sets of indices $\cI$ and $\cI'$ such that $\sum_{i \in \cI} \tr{\Pi_{S^{i}_n} \psi_n} > \sqrt \eps$ and $\sum_{i \in \cI'} \tr{\Pi_{S^{i}_n} \psi_n} > \sqrt \eps$.
However, this contradicts \cref{lem:local_clustering} since $\cup_{i \in \cI} S_n^i$ and $\cup_{i \in \cI'} S_n^i$ are separated by Hamming distance at least $\tilde \nu_2 n > k$.

The claim that $i$ is unique holds because if there were at least two such sets $S^i_n$ and $S^{i'}_n$, then $\tr{\Pi_{G_n} \psi_n} \geq \tr{\Pi_{S^{i}_n} \psi_n} + \tr{\Pi_{S^{i'}_n} \psi_n} \geq 2\,\tr{\Pi_{G_n} \psi_n} - 2 \sqrt \eps$, which is a contradiction since $\tr{\Pi_{G_n} \psi_n} \geq 4 \sqrt \eps$.
\end{proof}

\subsection{Limitations on symmetric optimisers for symmetric COPs}
\label{subsec:symopt}

From  \cref{lem:unique_cluster} we immediately obtain limitations on the performance of \emph{symmetric} local optimisers, i.e.~optimisers that are invariant under the operation $\sigma_X^{\ot n}$, on symmetric COPs $C_n(z)$, i.e.~COPs that satisfy $C_n(z) = C_n(- z)$.

\begin{corollary} \label{lem:symm_limits}
Suppose the objective $C_n(z)$ has the $(\mu, \nu_1, \nu_2)$-OGP with $2 \tilde \nu_1 < \tilde \nu_2$ and satisfies $C_n(z) = C_n(- z)$.
Then, no $(k, \eps)$-local quantum state $\psi_n = \proj{\psi_n}$ for $k = o(n)$ that satisfies $\sigma_X^{\ot n} \psi_n \sigma_X^{\ot n} = \psi_n$ can $(\mu, 1 - 4\sqrt \eps)$-optimise $C_n(z)$.
\end{corollary}
\begin{proof}
Suppose for the sake contradiction that such a local symmetric $\psi_n$ does $(\mu, 1 - 4\sqrt \eps)$-optimise $C_n(z)$.
Then it follows from \cref{lem:unique_cluster} that there exists a unique $S^i_n$ for which $\tr{\Pi_{S^{i}_n} \psi_n} \geq \tr{\Pi_{G_n}\psi_n} - \sqrt \eps$.
Now consider the set $S = \{x \oplus 1^n \;|\; x \in S_n^i\}$, where $1^n$ is the all-1 string.
By symmetry of $C_n(z)$ we have that for $x \in S$, $C_n((-1)^{x}) = C_n((-1)^{x \oplus 1^n}) \geq \mu$, where the inequality holds because $x \oplus 1^n \in S_n^i$.
Furthermore, it is easy to see that $|x - x'| \leq \tilde \nu_1 \cdot n$ for $x, x' \in S$.
Therefore, there must exist an $S_n^j$ in the partition of $G_n$ such that $S \subset S_n^j$. (In fact, by symmetry it is easy to see that $S = S_n^j$ for some $j$, but we will not need this here.)
Since there exist strings $x \in S_n^i$ and $x' \in S$ with Hamming distance $n$, we must have $i \neq j$.
Furthermore, by symmetry of $\psi_n$ and the fact that $\Pi_S = \sigma_X^{\ot n}\Pi_{S^i_n} \sigma_X^{\ot n}$, 
\begin{align*}
\tr{\Pi_{S^j_n} \psi_n} \geq \tr{\Pi_{S} \psi_n} = \tr{\Pi_{S^i_n} \sigma_X^{\ot n} \psi_n \sigma_X^{\ot n}} = \tr{\Pi_{S^i_n} \psi_n} \geq \tr{\Pi_{G_n}\psi_n} - \sqrt \eps \,.
\end{align*}
This contradicts the uniqueness of the high-weight set in \cref{lem:unique_cluster}.
\end{proof}

\subsection{Example: QAOA on symmetric COPs} \label{sec:symm_qaoa_limitations}

An example of practical relevance is the QAOA for approximately solving local COPs.
A $q$-local COPs and its associated $q$-local Hamiltonian can be written in terms of coefficients $J = (J_{i_1, \dots, i_q}) \in \R^{n^q}$ as\footnote{Note that because here we consider arbitrary coefficients $J$, the normalisation factors $\frac{1}{n^{(q+1)/2}}$ and $\frac{1}{n^{(q-1)/2}}$ can be chosen arbitrarily, too; we use these ones to keep the notation consistent with \cref{sec:sk_model}, but emphasise that here the coefficients need not be chosen from a Gaussian distribution.}
\begin{align*}
C_n^q(z; J) = \frac{1}{n^{(q+1)/2}}\sum_{i_1,\ldots i_q=1}^n J_{i_1,\ldots i_q} z_{i_1}\ldots z_{i_q} \,, \quad H_n^q(J) = \frac{1}{n^{(q-1)/2}} \sum_{i_1,\ldots i_q=1}^n J_{i_1,\ldots i_q} \sigma^Z_{i_1}\ldots \sigma^Z_{i_q} \,.
\end{align*}
The QAOA works by repeatedly applying the unitaries 
\begin{align*}
V_{H_n^q(J)}(\gamma) = e^{-\i \gamma H_n^q(J)}
\quad\tand\quad 
V_X(\beta) = \left( e^{-\i \beta \sigma_x}  \right)^{\ot n} 
\end{align*}
to the initial state $\rho_0 = \proj{+}^{\ot n}$ for some parameters $\gamma, \beta > 0$.
The output state of the QAOA can therefore be written as 
\begin{align}
\rho_{p,n}(J) = \Big( V_X(\beta_p) V_{H_n^q(J)}(\gamma_p) \cdots V_X(\beta_1) V_{H_n^q(J)}(\gamma_1) \Big) \rho_0  \Big( V_{H_n^q(J)}(\gamma_1)^\dagger  V_X(\beta_1)^\dagger  \cdots V_{H_n^q(J)}(\gamma_p)^\dagger  V_X(\beta_p)^\dagger \Big) \,. \label{eqn:qaoa_state}
\end{align}
The parameters $\gamma_i$ and $\beta_i$ may depend arbitrarily on $J$, but we will always assume that they are chosen from some fixed bounded range that is independent of $n$.

If $q$ is even, both the COP $C_n^q(z; J)$ and the QAOA output state $\rho_{p,n}$ are symmetric in the sense of \cref{lem:symm_limits}.
Therefore, if $C_n^q(z; J)$ satisfies the $(\mu, \nu_1, \nu_2)$-OGP (for some fixed value of $J$) and the associated Hamiltonian $H_n^q(J)$ satisfies the norm constraint from \cref{eqn:subset_condition}, we can combine \cref{thm:qaoamain} and \cref{lem:symm_limits} to obtain the following result, which shows that at level $p = o(\log\log n)$, the probability that the QAOA will produce a ``$\mu$-good'' string $x$ decays with $e^{-n^{1/8}}$.
\begin{lemma} \label{lem:symm_qaoa_limit}
Fix $q$ even.
Suppose that the COP $C_n^q(z; J)$ (for some choice of $J$) has the $(\mu, \nu_1, \nu_2)$-OGP with $2 \tilde \nu_1 < \tilde \nu_2$ and the Hamiltonian $H^q_n$ satisfies \cref{eqn:subset_condition}.
Then, no QAOA output state $\rho_{p,n}(J)$ (for any choice of $\gamma_i$ and $\beta_i$ within an arbitrary bounded range independent of $n$) can $(\mu, 1 - 8 e^{-n^{1/8}/\sqrt 8})$-optimise $C_n^q(z; J)$. 
In other words, if one measures the QAOA output state $\rho_{p,n}(J)$ with level $p = o(\log \log n)$ in the computational basis, the probability of receiving a string $x$ that satisfies $C_n^q((-1)^{x}; J) \geq \mu$ is at most $e^{-\Omega(n^{1/8})}$.
\end{lemma}
\begin{proof}
It is easy to verify that the QAOA for this $H_n^q(J)$ implements an instance of dense Hamiltonian evolution with $H^{(1)}=\gamma_1 H_n^q(J), H^{(2)}=\beta_1 \sum_i \sigma^{(i)}_x, H^{(3)}=\gamma_2 H_n^q(J)$, etc. 
Since $H_n^q(J)$ satisfies the condition in \cref{eqn:subset_condition} for some $\tilde C = O(1)$ and $\alpha < 1$, each $H^{(i)}$ also satisfies \cref{eqn:subset_condition} for $\tilde C \rightarrow \max (\tilde C \gamma_i, \beta_i) = O(1)$ and the same $\alpha$.
Therefore, we can apply \cref{thm:qaoamain} and find that for $p = o(\log \log n)$, $\rho_{p,n}(J)$ is a $(k, \eps)$-local operator for $k = o(n)$ and $\eps \leq 4 e^{-n^{1/8}/\sqrt 2}$.
Furthermore, since $C_n^q(z; J)$ is symmetric under $z \mapsto - z$ for even $q$ and any $J$, $H_n^q(J)$ commutes with $\sigma_x^{\ot n}$, so the QAOA output state $\rho_{p,n}(J)$ satisfies $\sigma_x^{\ot n} \rho_{p,n}(J) \sigma_x^{\ot n} = \rho_{p,n}(J)$. 
Therefore, the lemma follows from \cref{lem:symm_limits}.
\end{proof}

An example of a COP that satisfies the requirements of \cref{lem:symm_qaoa_limit} is the pure $q$-spin model from \cref{sec:sk_model}, for which we showed in \cref{lem:pure_spin} that \cref{eqn:subset_condition} is satisfied except with probability $e^{-n}$.
It is known that for fixed $q$, the limit $E_q(J) \deq \lim_{n\to\infty} \max_{z \in \pmset^n} C_n^q(z;J)$ exists almost surely~\cite{guerra2002thermodynamic}.
Furthermore, for even $q \geq 4$, with probability at least $1 - O(e^{-n})$ over the choice of $J$ (with i.i.d.~standard Gaussian entries), the objective function $C_n^q(z ; J)$ satisfies the $(\mu, \nu_1, \nu_2)$-OGP for constants $0 < \nu_1 < \nu_2 < 1$ and $0 < \mu < E_q(J)$~\cite{CGPR19,gamarnik2020low,gamarnik2021overlap}.
We note that~\cite{CGPR19} do not explicitly show that $2 \tilde \nu_1 < \tilde \nu_2$, although it appears to be implicit in the proof that $\tilde \nu_1$ can be made an arbitrarily small constant by choosing $\mu$ arbitrarily close to $E_q(J)$.
We leave a detailed proof of this statement for future work.
Therefore, we get the following implication of \cref{lem:symm_qaoa_limit}, giving the first provable (modulo the conjectured strengthened OGP) limitations on the QAOA on dense instances at super-constant level and improving upon the result of~\cite{BGMZ22}.
\begin{corollary} \label{cor:qaoa_spin_model}
Assuming the strengthened OGP stated above, with probability $1 - O(e^{-n})$ over the choice of $J$ (with i.i.d.~standard Gaussian entries), the value of the solution to the COP $C_n^q(z;J)$ produced by the QAOA with level $p = o(\log\log n)$ is bounded away from the optimal value by at least a constant except with probability $e^{-\Omega(n^{1/8})}$.
\end{corollary}
Using \cref{lem:mixed_spin} instead of \cref{lem:pure_spin}, this result can easily be extended to mixed spin models $C^{q, \setft{mixed}}_n(z; J) = \sum_{j = 1}^q c_j \, C^j_n(z; J)$ that only contain contributions from even pure spin models, i.e.~$c_j$ is non-zero only for even $j \geq 4$.

\bibliographystyle{alphaURL}

\bibliography{references}

\begin{thebibliography}{PGVWC07}

\bibitem[AAG22]{AAG21}
Anurag Anshu, Itai Arad, and David Gosset.
\newblock An area law for 2d frustration-free spin systems.
\newblock In {\em Proceedings of the 54th Annual ACM SIGACT Symposium on Theory
  of Computing}, STOC 2022, page 12–18, New York, NY, USA, 2022. Association
  for Computing Machinery.
\newblock \href {https://doi.org/10.1145/3519935.3519962}
  {\path{doi:10.1145/3519935.3519962}}.

\bibitem[AB22]{AB22}
Anurag Anshu and Nikolas~P. Breuckmann.
\newblock {A construction of combinatorial NLTS}.
\newblock {\em Journal of Mathematical Physics}, 63(12), 12 2022.
\newblock \href {https://doi.org/10.1063/5.0113731}
  {\path{doi:10.1063/5.0113731}}.

\bibitem[ABN22]{ABN22}
Anurag Anshu, Nikolas Breuckmann, and Chinmay Nirkhe.
\newblock {NLTS} hamiltonians from good quantum codes.
\newblock {\em arXiv preprint arXiv:2206.13228}, 2022.
\newblock URL: \url{https://arxiv.org/abs/2206.13228}.

\bibitem[Abr20]{Ab20}
Nilin Abrahamsen.
\newblock Short proof of a spectral chernoff bound for local hamiltonians.
\newblock {\em arXiv preprint arXiv:2009.04993}, 2020.
\newblock URL: \url{https://arxiv.org/abs/2009.04993}.

\bibitem[Alh22]{Al22}
{\'A}lvaro~M Alhambra.
\newblock Quantum many-body systems in thermal equilibrium.
\newblock {\em arXiv preprint arXiv:2204.08349}, 2022.
\newblock URL: \url{https://arxiv.org/abs/2204.08349}.

\bibitem[AN22]{ANirkhe22}
Anurag Anshu and Chinmay Nirkhe.
\newblock {Circuit Lower Bounds for Low-Energy States of Quantum Code
  Hamiltonians}.
\newblock In Mark Braverman, editor, {\em 13th Innovations in Theoretical
  Computer Science Conference (ITCS 2022)}, volume 215 of {\em Leibniz
  International Proceedings in Informatics (LIPIcs)}, pages 6:1--6:22,
  Dagstuhl, Germany, 2022. Schloss Dagstuhl -- Leibniz-Zentrum f{\"u}r
  Informatik.
\newblock \href {https://doi.org/10.4230/LIPIcs.ITCS.2022.6}
  {\path{doi:10.4230/LIPIcs.ITCS.2022.6}}.

\bibitem[Ans16]{A16}
Anurag Anshu.
\newblock Concentration bounds for quantum states with finite correlation
  length on quantum spin lattice systems.
\newblock {\em New Journal of Physics}, 18(8):083011, aug 2016.
\newblock \href {https://doi.org/10.1088/1367-2630/18/8/083011}
  {\path{doi:10.1088/1367-2630/18/8/083011}}.

\bibitem[BC15]{BC15}
Fernando~GSL Brandao and Marcus Cramer.
\newblock Equivalence of statistical mechanical ensembles for non-critical
  quantum systems.
\newblock {\em arXiv preprint arXiv:1502.03263}, 2015.
\newblock URL: \url{https://arxiv.org/abs/1502.03263}.

\bibitem[BCDZ99]{BuhrmanCWZ99}
H.~{Buhrman}, R.~{Cleve}, R.~{De Wolf}, and C.~{Zalka}.
\newblock Bounds for small-error and zero-error quantum algorithms.
\newblock In {\em Proceedings of the 40th Annual Symposium on Foundations of
  Computer Science}, pages 358--368, 1999.
\newblock \href {https://doi.org/10.1109/SFFCS.1999.814607}
  {\path{doi:10.1109/SFFCS.1999.814607}}.

\bibitem[BCG15]{BCG15}
FGSL Brandao, Marcus Cramer, and Madalin Guta.
\newblock A berry--esseen theorem for quantum lattice systems and the
  equivalence of statistical mechanical ensembles.
\newblock {\em QIP2015 Talk}, 2015.
\newblock URL: \url{http://www.quantum-lab.org/qip2015/talks/125-Brandao.pdf}.

\bibitem[BGMZ22]{BGMZ22}
Joao Basso, David Gamarnik, Song Mei, and Leo Zhou.
\newblock Performance and limitations of the qaoa at constant levels on large
  sparse hypergraphs and spin glass models.
\newblock In {\em 2022 IEEE 63rd Annual Symposium on Foundations of Computer
  Science (FOCS)}, 2022.
\newblock \href {https://doi.org/10.1109/FOCS54457.2022.00039}
  {\path{doi:10.1109/FOCS54457.2022.00039}}.

\bibitem[BKKT]{BravyiKKT19}
Sergey Bravyi, Alexander Kliesch, Robert Koenig, and Eugene Tang.
\newblock Obstacles to variational quantum optimization from symmetry
  protection.
\newblock {\em Phys. Rev. Lett.}, 125:260505, Dec.
\newblock \href {https://doi.org/10.1103/PhysRevLett.125.260505}
  {\path{doi:10.1103/PhysRevLett.125.260505}}.

\bibitem[CGPR19]{CGPR19}
Wei-Kuo Chen, David Gamarnik, Dmitry Panchenko, and Mustazee Rahman.
\newblock Suboptimality of local algorithms for a class of max-cut problems.
\newblock {\em The Annals of Probability}, 47(3):1587 -- 1618, 2019.
\newblock \href {https://doi.org/10.1214/18-AOP1291}
  {\path{doi:10.1214/18-AOP1291}}.

\bibitem[CLSS22]{CLSS22}
Chi-Ning Chou, Peter~J. Love, Juspreet~Singh Sandhu, and Jonathan Shi.
\newblock {Limitations of Local Quantum Algorithms on Random MAX-k-XOR and
  Beyond}.
\newblock In {\em 49th International Colloquium on Automata, Languages, and
  Programming (ICALP 2022)}, volume 229 of {\em Leibniz International
  Proceedings in Informatics (LIPIcs)}, pages 41:1--41:20, Dagstuhl, Germany,
  2022. Schloss Dagstuhl -- Leibniz-Zentrum fur Informatik.
\newblock \href {https://doi.org/10.4230/LIPIcs.ICALP.2022.41}
  {\path{doi:10.4230/LIPIcs.ICALP.2022.41}}.

\bibitem[CPGSV21]{cirac2021matrix}
J~Ignacio Cirac, David Perez-Garcia, Norbert Schuch, and Frank Verstraete.
\newblock Matrix product states and projected entangled pair states: Concepts,
  symmetries, theorems.
\newblock {\em Reviews of Modern Physics}, 93(4):045003, 2021.
\newblock \href {https://doi.org/10.1103/RevModPhys.93.045003}
  {\path{doi:10.1103/RevModPhys.93.045003}}.

\bibitem[DPMRF23]{DMRF22}
Giacomo De~Palma, Milad Marvian, Cambyse Rouz\'e, and Daniel~Stilck Franca.
\newblock Limitations of variational quantum algorithms: A quantum optimal
  transport approach.
\newblock {\em PRX Quantum}, 4:010309, Jan 2023.
\newblock \href {https://doi.org/10.1103/PRXQuantum.4.010309}
  {\path{doi:10.1103/PRXQuantum.4.010309}}.

\bibitem[DPR22]{DR22}
Giacomo De~Palma and Cambyse Rouz{\'e}.
\newblock Quantum concentration inequalities.
\newblock {\em Annales Henri Poincar{\'e}}, Apr 2022.
\newblock \href {https://doi.org/10.1007/s00023-022-01181-1}
  {\path{doi:10.1007/s00023-022-01181-1}}.

\bibitem[EH17]{EldarH17}
L.~{Eldar} and A.~W. {Harrow}.
\newblock Local hamiltonians whose ground states are hard to approximate.
\newblock In {\em 2017 IEEE 58th Annual Symposium on Foundations of Computer
  Science (FOCS)}, pages 427--438, 2017.
\newblock \href {https://doi.org/10.1109/FOCS.2017.46}
  {\path{doi:10.1109/FOCS.2017.46}}.

\bibitem[FGG14]{farhi2014quantum}
Edward Farhi, Jeffrey Goldstone, and Sam Gutmann.
\newblock A quantum approximate optimization algorithm.
\newblock {\em arXiv preprint arXiv:1411.4028}, 2014.
\newblock URL: \url{https://arxiv.org/abs/1411.4028}.

\bibitem[FGG20]{farhi2020quantum}
Edward Farhi, David Gamarnik, and Sam Gutmann.
\newblock The quantum approximate optimization algorithm needs to see the whole
  graph: A typical case.
\newblock {\em arXiv preprint arXiv:2004.09002}, 2020.
\newblock URL: \url{https://arxiv.org/abs/2004.09002}.

\bibitem[Gam21]{gamarnik2021survey}
David Gamarnik.
\newblock The overlap gap property: A topological barrier to optimizing over
  random structures.
\newblock {\em Proceedings of the National Academy of Sciences},
  118(41):e2108492118, 2021.
\newblock \href {https://doi.org/10.1073/pnas.2108492118}
  {\path{doi:10.1073/pnas.2108492118}}.

\bibitem[GGL95]{graham1995handbook}
Ronald~L Graham, Martin Gr{\"o}tschel, and L{\'a}szl{\'o} Lov{\'a}sz.
\newblock {\em Handbook of Combinatorics}, volume~1.
\newblock Elsevier, 1995.

\bibitem[GJ21]{gamarnik2021overlap}
David Gamarnik and Aukosh Jagannath.
\newblock The overlap gap property and approximate message passing algorithms
  for $p$-spin models.
\newblock {\em The Annals of Probability}, 49(1):180--205, 2021.
\newblock URL: \url{https://hdl.handle.net/1721.1/145311}.

\bibitem[GJW20]{gamarnik2020low}
David Gamarnik, Aukosh Jagannath, and Alexander~S Wein.
\newblock Low-degree hardness of random optimization problems.
\newblock In {\em 2020 IEEE 61st Annual Symposium on Foundations of Computer
  Science (FOCS)}, pages 131--140. IEEE, 2020.
\newblock \href {https://doi.org/10.1109/FOCS46700.2020.00021}
  {\path{doi:10.1109/FOCS46700.2020.00021}}.

\bibitem[GL18]{gamarnik2018finding}
David Gamarnik and Quan Li.
\newblock Finding a large submatrix of a gaussian random matrix.
\newblock {\em The Annals of Statistics}, 46(6A):2511--2561, 2018.
\newblock URL: \url{http://hdl.handle.net/1721.1/120593}.

\bibitem[GT02]{guerra2002thermodynamic}
Francesco Guerra and Fabio~Lucio Toninelli.
\newblock The thermodynamic limit in mean field spin glass models.
\newblock {\em Communications in Mathematical Physics}, 230(1):71--79, 2002.
\newblock \href {https://doi.org/10.1007/s00220-002-0699-y}
  {\path{doi:10.1007/s00220-002-0699-y}}.

\bibitem[GV89]{GV89}
D.~Goderis and P.~Vets.
\newblock Central limit theorem for mixing quantum systems and the
  {CCR}-algebra of fluctuations.
\newblock {\em Communications in Mathematical Physics}, 122:249--265, 1989.
\newblock \href {https://doi.org/10.1007/BF01257415}
  {\path{doi:10.1007/BF01257415}}.

\bibitem[Has04]{Hastings04}
M.~B. Hastings.
\newblock Lieb-schultz-mattis in higher dimensions.
\newblock {\em Phys. Rev. B}, 69:104431, Mar 2004.
\newblock \href {https://doi.org/10.1103/PhysRevB.69.104431}
  {\path{doi:10.1103/PhysRevB.69.104431}}.

\bibitem[HMH04]{HGH04}
Michael Hartmann, G\"unter Mahler, and Ortwin Hess.
\newblock Existence of {T}emperature on the {N}anoscale.
\newblock {\em Phys. Rev. Lett.}, 93:080402, Aug 2004.
\newblock \href {https://doi.org/10.1103/PhysRevLett.93.080402}
  {\path{doi:10.1103/PhysRevLett.93.080402}}.

\bibitem[KAAV17]{KAAV15}
Tomotaka Kuwahara, Itai Arad, Luigi Amico, and Vlatko Vedral.
\newblock Local reversibility and entanglement structure of many-body ground
  states.
\newblock {\em Quantum Science and Technology}, 2(1):015005, 2017.
\newblock \href {https://doi.org/10.1088/2058-9565/aa523d}
  {\path{doi:10.1088/2058-9565/aa523d}}.

\bibitem[KLS96]{KahnLS96}
Jeff Kahn, Nathan Linial, and Alex Samorodnitsky.
\newblock Inclusion-exclusion: Exact and approximate.
\newblock {\em Combinatorica}, 16(4):465--477, Dec 1996.
\newblock \href {https://doi.org/10.1007/BF01271266}
  {\path{doi:10.1007/BF01271266}}.

\bibitem[KS20a]{KS20b}
Tomotaka Kuwahara and Keiji Saito.
\newblock Eigenstate thermalization from the clustering property of
  correlation.
\newblock {\em Phys. Rev. Lett.}, 124:200604, May 2020.
\newblock \href {https://doi.org/10.1103/PhysRevLett.124.200604}
  {\path{doi:10.1103/PhysRevLett.124.200604}}.

\bibitem[KS20b]{KS20}
Tomotaka Kuwahara and Keiji Saito.
\newblock Gaussian concentration bound and ensemble equivalence in generic
  quantum many-body systems including long-range interactions.
\newblock {\em Annals of Physics}, 421:168278, 2020.
\newblock \href {https://doi.org/10.1016/j.aop.2020.168278}
  {\path{doi:10.1016/j.aop.2020.168278}}.

\bibitem[Kuw16]{K16}
Tomotaka Kuwahara.
\newblock Connecting the probability distributions of different operators and
  generalization of the chernoff-hoeffding inequality.
\newblock {\em Journal of Statistical Mechanics: Theory and Experiment},
  2016(11):113103, nov 2016.
\newblock \href {https://doi.org/10.1088/1742-5468/2016/11/113103}
  {\path{doi:10.1088/1742-5468/2016/11/113103}}.

\bibitem[LSM61]{LiebSM61}
Elliott Lieb, Theodore Schultz, and Daniel Mattis.
\newblock Two soluble models of an antiferromagnetic chain.
\newblock {\em Annals of Physics}, 16(3):407--466, 1961.
\newblock \href {https://doi.org/10.1016/0003-4916(61)90115-4}
  {\path{doi:10.1016/0003-4916(61)90115-4}}.

\bibitem[LZ22]{quantum-tanner-codes}
Anthony Leverrier and Gilles Z{\'e}mor.
\newblock Quantum tanner codes.
\newblock In {\em 2022 IEEE 63rd Annual Symposium on Foundations of Computer
  Science (FOCS)}, 2022.
\newblock \href {https://doi.org/10.1109/FOCS54457.2022.00117}
  {\path{doi:10.1109/FOCS54457.2022.00117}}.

\bibitem[PGVWC07]{PVWC07}
D.~Perez-Garcia, F.~Verstraete, M.~M. Wolf, and J.~I. Cirac.
\newblock Matrix product state representations.
\newblock {\em Quantum Info. Comput.}, 7(5):401–430, jul 2007.
\newblock URL: \url{https://dl.acm.org/doi/10.5555/2011832.2011833}.

\bibitem[PK22]{panteleevk2021}
Pavel Panteleev and Gleb Kalachev.
\newblock Asymptotically good quantum and locally testable classical ldpc
  codes.
\newblock In {\em Proceedings of the 54th Annual ACM SIGACT Symposium on Theory
  of Computing}, pages 375--388, 2022.
\newblock \href {https://doi.org/10.1145/3519935.3520017}
  {\path{doi:10.1145/3519935.3520017}}.

\bibitem[SV14]{SV14}
Sushant Sachdeva and Nisheeth~K. Vishnoi.
\newblock Faster algorithms via approximation theory.
\newblock {\em Foundations and Trends® in Theoretical Computer Science},
  9(2):125--210, 2014.
\newblock \href {https://doi.org/10.1561/0400000065}
  {\path{doi:10.1561/0400000065}}.

\bibitem[Tas18]{T18}
Hal Tasaki.
\newblock On the local equivalence between the canonical and the microcanonical
  ensembles for quantum spin systems.
\newblock {\em Journal of Statistical Physics}, 172(4):905--926, Aug 2018.
\newblock \href {https://doi.org/10.1007/s10955-018-2077-y}
  {\path{doi:10.1007/s10955-018-2077-y}}.

\end{thebibliography}

\appendix

\section{Locality spread under Hamiltonian evolution}
\begin{lemma} \label{lem:gen_locality_spread}
Let $O$ be an operator that acts non-trivially on at most $k$ qubits and $H$ an $\ell$-local Hamiltonian.
Then $e^{-\i H} O e^{\i H}$ is a $(k', \eps')$-local operator for 
\begin{align*}
k' = 2 \ell d + k \,, \quad \eps' = 3 e^{-(d - e \norm{H})} \norm{O} \,,
\end{align*}
for any integer $d \geq e \norm{H}$.
\end{lemma}

\begin{proof}
Let 
\begin{align*}
Q = \1 + \sum_{m = 1}^d \frac{(-\i H)^m}{m!}
\end{align*}
be the Taylor series of $e^{-\i H}$ truncated at degree $d$.
By Taylor's theorem, $\norm{e^{-\i H} - Q} \leq e^{-(d - e \norm{H})}$, so it follows from the triangle inequality and submultiplicativity of the norm that for $d \geq e \norm{H}$,
\begin{align*}
\norm{e^{-\i H} O e^{\i  H} - Q O Q^\dagger } \leq 3 e^{-(d - e \norm{H})} \norm{O} \,.
\end{align*}
Because $H^d$ can be expanded as a sum of terms that each contain at most $d$ of the $\ell$-local terms $H_i$, we see that $Q$ is $(\ell d)$-local.
Therefore, $Q O Q^\dagger$ is $(2 \ell d + k)$-local, concluding the proof.
\end{proof}

\begin{lemma} \label{lem:comm_locality_spread}
Let $R = \sum_{i} R_i$ be a $k$-local operator and $H = \sum H_i$ an $\ell$-local \emph{commuting} Hamiltonian.
Suppose that there exist constants $\alpha \in [0, 1)$ and $\tilde C > 0$ such that for every subset $S \subseteq [n]$ of qubits, the subset Hamiltonian $H_S$ (see \cref{def:subset_op}) satisfies $\norm{H_S} \leq \tilde C n^\alpha |S|^{1 - \alpha}$.
Then, $e^{-\i H} R e^{\i H}$ is a $(k', \eps')$-local operator for
\begin{align}
k' = 2 \ell \lceil \mu \tilde C e n^\alpha \kappa^{1 - \alpha} \rceil + k \,, \quad 
\eps' = 3 e^{-(\mu - 1) \tilde C e n^\alpha \kappa^{1 - \alpha}} \tln(R) \,, \label{eqn:ham_locality_bound}
\end{align}
for any $\mu > 1$ and $\kappa \geq k$.
Furthermore, 
\begin{align}
\tln_{\eps'}\left( e^{-\i H} R e^{\i H} \right) \leq (2 n)^{k'} \left( \tln(R) + \eps' \right) \,. \label{eqn:norm_bound}
\end{align}
\end{lemma}

\begin{proof}
For each $i \in [t]$, we define $S_i$ as the subset of qubits on which $R_i$ acts non-trivially.
Since $R$ is $k$-local, $|S_i| \leq k$ for all $i$.
Because $H$ is commuting, 
\begin{align*}
e^{- \i H} R_i e^{\i H} = e^{- \i H_{S_i}} R_i e^{\i H_{S_i}}\,.
\end{align*}
Recall that $\norm{H_{S_i}} \leq \tilde C n^\alpha k^{1 - \alpha} \leq \tilde C n^\alpha \kappa^{1 - \alpha}$.
We can therefore apply \cref{lem:gen_locality_spread} with $d = \lceil\mu \tilde C e n^\alpha \kappa^{1 - \alpha}\rceil > e \norm{H_{S_i}}$ to find that $e^{- \i H_{S_i}} R_i e^{\i H_{S_i}}$ is a $(k', \eps'_i)$-local operator with 
\begin{align*}
k' = 2 \ell \lceil \mu \tilde C e n^\alpha \kappa^{1 - \alpha} \rceil + k \,, \quad 
\eps'_i = 3 e^{-(\mu - 1) \tilde C e n^\alpha \kappa^{1 - \alpha}} \norm{R_i} \,.
\end{align*}
We define $\tilde R^i$ as the local operator approximations to $e^{- \i H_{S_i}} R_i e^{\i H_{S_i}}$ (as given by \cref{lem:gen_locality_spread}) and $\tilde R = \sum_i \tilde R^i$.
(We use superscripts for $\tilde R^i$ because each $\tilde R^i$ is itself a $k'$-local operator, not an operator that acts non-trivially on only $k'$ qubits, i.e.~$\tilde R = \sum_i \tilde R^i$ is not our usual local decomposition.)
Since the sum of $k'$-local operators is still $k'$-local, we see that $\tilde R$ is $k'$-local operator and, by the triangle inequality, approximates $e^{-\i H} R e^{\i H} = \sum_i e^{-\i H_{S_i}} R_i e^{\i H_{S_i}}$ to within error 
\begin{align*}
\eps' \leq \sum_{i} \eps'_i = 3 e^{-(\mu - 1) \tilde C e n^\alpha \kappa^{1 - \alpha}} \sum_i \norm{R_i} = 3 e^{-(\mu - 1) \tilde C e n^\alpha \kappa^{1 - \alpha}} \tln(R) \,.
\end{align*}
This completes the proof of \cref{eqn:ham_locality_bound}.

To show \cref{eqn:norm_bound}, we first bound $\tln(\tilde R^i)$ in terms of $\norm{\tilde R^i}$.
Since $\tilde R^i$ is $k'$-local, we can expand 
\begin{align*}
\tilde R^i = \sum_{T \subseteq [n] \sth |T| \leq k'} \tilde R^i_{(T)}
\end{align*} 
using the unique decomposition from \cref{eqn:unique_decomp}.
The number of terms $\tilde R^i_{(T)}$ is at most $\sum_{j = 0}^{k'} {n \choose j} \leq n^{k'}$.
Therefore, 
\begin{align*}
\tln(\tilde R^i) 
\leq \sum_T \norm{\tilde R^i_{(T)}}
\leq n^{k'} \max_T \norm{\tilde R^i_{(T)}}
\leq n^{k'} \cdot 2^{k'}\cdot \norm{\tilde R^i}\,,
\end{align*}
where the last inequality follows from \cref{lem:opbound}.
Due to unitary invariance of the norm:
\begin{align*}
\norm{\tilde R^i} 
&\leq \norm{e^{- \i H_{S_i}} R_i e^{\i H_{S_i}}} + \eps'_i = \norm{R_i} + \eps'_i 
\end{align*}
Finally, we combine the above bounds to obtain \cref{eqn:norm_bound}: 
\begin{align*}
\tln_{\eps'} \left( e^{-\i H} R e^{\i H} \right) 
\leq \sum_{i} \tln(\tilde R^i) 
\leq (2 n)^{k'} \left( \tln(R) + \eps' \right) \,. 
\end{align*}
\end{proof}

\end{document}